\documentclass[arunningheads,a4paper]{llncs}

\usepackage{amssymb,xspace,latexsym,epic,eepic,graphicx}
\usepackage{epsfig,amsmath}
\usepackage{colortbl}

\usepackage{epstopdf}

\usepackage[boxruled,vlined]{algorithm2e}

\usepackage{xcolor}

\usepackage{tikz}
\usetikzlibrary{arrows}

\newcommand{\ryc}[1]{}
\newcommand{\jd}[1]{}
\newcommand{\jr}[1]{}

\newcommand{\Nulls}{\mathcal{N}}

\newcommand{\nulls}{\textit{Nulls}}

\renewcommand{\epsilon}{\varepsilon}
\newcommand{\Sch}{{\cal S}}
\newcommand{\Inst}{{\cal I}}

\newcommand{\schema}{\text{Schema}}

\newcommand{\Pro}{P}

\newcommand{\C}{\mathcal{C}}
\newcommand{\Const}{\mathbb{C}}
\newcommand{\T}{\mathbb{T}}
\renewcommand{\P}{\mathbb{P}}

\newcommand{\Que}{\mathbb{Q}}
\newcommand{\Q}{\mathcal{Q}}
\newcommand{\pre}{\text{in}}
\newcommand{\post}{\text{out}}
\newcommand{\safe}{\text{safe}}
\newcommand{\Rel}{\textit{Rel}}

\newcommand{\outcome}{\textit{outcome}}

\newcommand{\rep}{\textit{rep}}

\newcommand{\Cpre}{\mathcal{C}_\text{in}}
\newcommand{\Cpost}{\mathcal{C}_\text{out}}
\newcommand{\Scope}{\textit{Scope}}

\newcommand{\evisits}{\textit{EVisits}\xspace}
\newcommand{\locvisits}{\textit{LocVisits}\xspace}
\newcommand{\patients}{\textit{Patients}\xspace}

\newcommand{\fid}{\textit{facility}\xspace}
\newcommand{\pid}{\textit{patInsur}\xspace}
\newcommand{\timestamp}{\textit{timestp}\xspace}
\newcommand{\age}{\textit{age}\xspace}
\newcommand{\val}{\textit{val}\xspace}

\newcommand{\visits}{\textit{Visits}\xspace}
\newcommand{\floc}{\textit{facilityLoc}\xspace}
\newcommand{\name}{\textit{name}\xspace}
\newcommand{\zip}{\textit{zipCode}\xspace}
\newcommand{\insurinfo}{\textit{InsurerInfo}\xspace} 

\begin{document}

\mainmatter  

\title{A Framework for Assessing Achievability \\ Of Data-Quality Constraints}

\titlerunning{A Framework for Assessing Achievability of Data-Quality Constraints}

%
%
\author{Rada Chirkova$^{1}$\and Jon Doyle$^{1}$\and Juan L. Reutter$^{2}$}
\authorrunning{A Framework for Assessing Achievability of Data-Quality Constraints}

\institute{$^{1}$ Computer Science Department, North Carolina State University\\
North Carolina, USA\\
$^{2}$ Pontificia Universidad Cat\'olica de Chile\\
\texttt{chirkova@csc.ncsu.edu, Jon\_Doyle@ncsu.edu, jreutter@ing.puc.cl}  \\
}

%
%

\toctitle{Lecture Notes in Computer Science}
\tocauthor{Authors' Instructions}
\maketitle

\begin{abstract} 
Assessing and improving the quality of data are fundamental challenges for data-intensive systems that have given rise to numerous applications targeting transformation and cleaning of data. However, while schema design, data cleaning, and data migration are nowadays reasonably well understood in isolation, not much attention has been given to the interplay between the tools addressing issues in these areas. We focus on the problem of determining whether the available data-processing procedures can be used together to bring about the desired quality characteristics of the given data. For an illustration, consider an organization that is introducing new data-analysis tasks. Depending on the tasks, it may be a priority for the organization to determine whether its data can be processed and transformed using the available data-processing tools to satisfy certain properties or quality assurances needed for the success of the task. Here, while the organization may control some of its tools, some other tools may be external or proprietary, with only basic information available on how they process data. The problem is then, how to decide which tools to apply, and in which order, to make the data ready for the new tasks? 

Toward addressing this problem, we develop a new framework that abstracts data-processing tools as black-box procedures with only some of the properties exposed, such as the applicability requirements, the parts of the data that the procedure modifies, and the conditions that the data satisfy once the procedure has been applied. We show how common database tasks such as data cleaning and data migration are encapsulated into our framework and, as a proof of concept, we study basic properties of the framework for the case of procedures described by standard relational constraints. We show that, while reasoning in this framework may be computationally infeasible in general, there exist well-behaved special cases with potential practical applications. 
\end{abstract} 

\section{Introduction} 
\label{intro-sec}

A common approach to ascertaining and improving the quality of data 
is to develop procedures and workflows for repairing or improving data sets 
with respect to quality constraints. 
The community has identified a wide range of data-management 
problems that are vital in this respect, leading to the creation of several lines of studies, which 
have normally been followed by the development of toolboxes of 
applications that practitioners can use to solve their problems. This has been the case, for example, 
for the 
Extract-Transform-Load (ETL) \cite{Devlin97,Kimball04} process in
business applications, or for the development of automatic tools to reason about the completeness or cleanliness 
of the data \cite{2012FanGeertsBook}. 


As a result, organizations facing data-improvement problems now have access to a variety of data-management 
tools to choose from; the tools can be assembled to construct so-called workflows of data operations. 
However, in contrast with the considerable body of research on particular data operations, or even 
entire business workflows (see, e.g., \cite{DHPV09,BGHLS07,2012Deutsch,BerardiCGHLM05}), previous research appears to have not
focused explicitly either on the assembly process itself or on
providing guarantees that the desired data-quality constraints will be
satisfied once the assembled workflow of procedures has been applied
to the available data.  

We investigate the problem of constructing workflows from already available procedures. 
That is, we consider a scenario in which an organization needs to meet a certain data-quality criterion or goal using 
available data-improvement procedures. In this case, the problem is  to understand whether these procedures can be 
assembled into a data-improvement workflow in a way that would guarantee that the data set produced by the 
workflow will effectively meet the desired quality goal.

 \medskip
\noindent
\textbf{Motivating example}: 
Suppose that data stored in a medical-data aggregator (such
as, e.g., Premier \cite{premierWebSite}) are accessed to perform a
health-outcomes analysis in population health management
\cite{KindigS03,McAlearney03,Young98}, focusing on repeat
emergency-room visits in the Washington, DC area.  
The goal of the analysis is
to see whether there is a relationship between such repeat visits and
ages and zip codes of patients. 

We assume that the aggregator imports information about emergency-room
visits from a number of facilities, and stores the information using a
relation \visits with attributes \fid and \floc for the ID and location of the medical facility, 
\pid for the patient insurance number, and
\timestamp for the date and time of the visit. 
We also assume that medical-record information imported from each facility is stored at the aggregator in a relation 
\patients, with attributes \fid, \pid, \name, \age, \zip, and so on. 

The analyst plans to isolate information about emergency-room visits for the Washington area in a relation 
\locvisits, which would have all the attributes of \visits except \floc, as the values of the latter are understood to be fixed. 
Further, to obtain the age and zip code of patients, the analyst also needs to integrate 
the data in \visits with those of \patients. 

To process the data, the analyst has access to some procedures that are part of the aggregator's everyday business. 
For example, the aggregator periodically
runs a \textit{StandardizePatientInfo} procedure, which first
performs entity resolution on insurance IDs in $\patients$,
using both the values of all the patient-related attributes in that
relation and a separate ``master'' relation $\insurinfo$ that stores authoritative
patient information from insurance companies, and then merges the
results into $\visits$.  
Further, the aggregator offers a procedure \textit{MigrateIntoLocVisits} that will directly populate \locvisits with the 
relevant information about emergency rooms (but not the age and zip code of patients). 

The analyst is now facing a number of choices, some of which we list here: 
\begin{itemize}
\item[(i)] Use the \textit{StandardizePatientInfo} procedure on
  \patients, then manually import the correct(ed) information
  into \locvisits, and finally join this relation with \patients. 
\item[(ii)] Run \textit{MigrateIntoLocVisits} to get the relevant patient information, and then join with \patients 
without running the procedure \textit{StandardizePatientInfo}. 
\item[(iii)] Add \age and \zip attributes to \locvisits, get the information into 
\locvisits as in (ii), and then try to modify \textit{StandardizePatientInfo} into operating directly on \locvisits. 
\end{itemize}

Which of these options is the best for the planned analysis? Option (i) seems to be the cleanest, but if the analyst 
suspects that \textit{StandardizePatientInfo} may produce some loss of data, then going with (ii) or (iii) might be a better option.  
Further, suppose the analyst also has access to a separate relation \textit{HealthcareInfo} from a health NGO, with 
information about emergency-room visits gathered from other independent sources. Then the analyst could pose the 
following quality criterion on the assembled workflow: The result of the workflow should provide at least the information that 
can be obtained from the relation \textit{HealthcareInfo}. How could one guarantee that such a criterion will be met?

 \medskip
\noindent
\textbf{Contributions}: 
Our goal is to develop a general framework that can be used to determine whether the available data-processing tools can be
put together into a workflow capable of producing data that meet the desired quality properties. 
To address this problem, we abstract data-processing tools as black-box procedures that  expose only certain properties. The properties of interest include (i) preconditions, which indicate the state of the data required for the procedure to be applicable; (ii) the parts of the data that the procedure modifies;  
and (iii) postconditions, which the data satisfy once the procedure has been applied. 

In this paper we introduce the basic building blocks and basic results for the proposed framework for assessing achievability of data-quality constraints. 
The contributions include formalizing the notion of (sequences of) data-transforming procedures, and characterizing instances that are outcomes of  applying  
(sequences of) procedures over other instances. We also illustrate our design choices by discussing ways to encode important database tasks in the proposed framework, including data migration, data cleaning, and schema updates.

One of the advantages of our framework is its generality, as it can be used to encode multiple operations not only on relational data, but on semistructured or even 
unstructured text data. 
This very generality implies that  to be able to reason about the properties of our framework, one needs to first instantiate some of its most abstract components. 
As a proof of concept, we provide an in-depth analysis of applications of (sequences of) procedures over relational data, where the procedures are stated using 
standard relational-data formalisms. 
We show that properties concerning outcomes of procedures are in general (not surprisingly) undecidable. At the same time, we achieve decidability and tractability for broad classes of realistic procedures 
that we illustrate with examples. While the formalism and results presented in this paper have practical implications on their own, we see them mainly as prerequisites that need to be understood before one can formalize the notion 
of assembling procedures in the context of and in response to a user task. We conclude this paper by showing how the proposed framework can be used to formally define the following problem: 
Given a set of procedures and data-quality criteria, is it possible to assemble a sequence of procedures such that the data outcome is assured to satisfy this criteria?


 \medskip
\noindent
\textbf{Related Work}: 
Researchers have been working on eliciting and defining specific dimensions of quality of the data --- \cite{WangS96} provides a widely acknowledged standard; please also see \cite{KahnSW02,LeeSKW02}. At the general level, high-quality data can be regarded as being fit for their intended use \cite{LeePWF09,WangLPS98,ChengalurSP98} --- that is, both context and use (i.e., tasks to be performed) need to be taken into account when evaluating and improving the quality of data. 
Recent efforts have put an emphasis on information-quality policies and strategies; please see \cite{LeePWF09} for a groundbreaking set of generic information-quality policies that structure decisions on information. An information-quality improvement cycle, consisting of the define-measure-analyze-improve steps for data quality, has been proposed in \cite{Wang98}. Work has also been done \cite{LeeS04} in the direction of integrating process measures with information-quality measures. Our work is different from these lines of research in that in our framework we assume that task-oriented data-quality requirements are already given in the form of constraints that need to be satisfied  on the data, and that procedures for improving data quality are also specified and available. Under these assumptions, our goal is to determine whether the procedures can be used to achieve satisfaction of the quality requirements on the data. 

The work \cite{2012FanGeertsBook} introduces a unified 
framework covering formalizations and approaches for a range of problems in data extraction, cleaning, repair, and integration, and also supplies an excellent survey of related work in these areas. More recent work on data cleaning includes \cite{BergmanMNTsigmod15,BergmanMNT15,KrishnanWFGKM015,RazniewskiKNS15,NuttPS15}. The research area of business processes \cite{2012Deutsch} studies the environment in which data are generated and transformed, including processes, users of data, and goals of using the data.  
In this context, researchers have studied automating composition of services into business processes, see, e.g.,  \cite{BerardiCGHLM05,BerardiCGHM05,BerardiCGLM05}, under the assumption that the assembly needs to follow a predefined  workflow of executions of actions (services). In contrast, in our work, the predefined part is the specified constraints that the data should satisfy after the assembled workflow of available procedures has been applied to it. Another line of work  \cite{DHPV09,BGHLS07} is closer to reasoning about static properties of business process workflows. That work is different from ours in that it does not pursue the goal of constructing new workflows. 

 \smallskip
\noindent
\textbf{Outline of the paper}: Section \ref{ref-prelim} contains basic definitions used  
in the paper. Section \ref{ref-procedures} introduces the proposed framework, and Section \ref{sec-exa} discusses 
encoding tasks such as data exchange, data cleaning, and alter-table statements. The formal results are presented in Section \ref{ref-basics}. Section \ref{ref-conc} concludes with a discussion of future challenges and opportunities. 



\section{Preliminaries}
\label{ref-prelim}

\noindent
\textbf{Schemas and instances}: Assume a countably infinite set of attribute names $\mathcal{A} = \{ A_1$, $A_2$, $\ldots \}$ and a countably infinite set (disjoint from $\cal A$) of relation names $\mathcal{R} = \{ R_1$, $R_2$, $\ldots \}$. 
A relational schema is a partial function $\Sch: \mathcal{R} \to 2^{\mathcal{A}}$ with finite domain, 
which associates a finite set of attributes to a finite set of relation symbols. If $\Sch(R)$ is defined, we say that $R$ is in $\Sch$.
A schema $\Sch'$ extends a schema $\Sch$ if for each relation $R$ such that $\Sch(R)$ is defined, we have that 
$\Sch(R) \subseteq \Sch'(R)$. That is, $\Sch'$ extends $\Sch$ if $\Sch'$ assigns at least the same attributes to all relations in $\Sch$.
We also assume a total order $\leq_\mathcal{A}$ over all attribute
names in order to be able to switch between the \emph{named} and
\emph{unnamed} perspectives for instances and queries.



We define instances so that it is possible to switch between the named and unnamed perspectives. 
Assume a countably infinite set of domain values $D$ (disjoint from both $\mathcal A$ and $\mathcal R$). 
Following \cite{AHV95}, an instance $I$ of schema $\Sch$ assigns to each relation $R$ in $\Sch$, 
where $\Sch(R) = \{A_1,\dots,A_n\}$, a set $R^I$ of \emph{named} tuples, each of which is a function of the form $t:\{A_1,\dots,A_n\} \rightarrow D$, 
representing the tuples in $R$. (We use $t(A_i)$ to denote the element of a tuple $t$ corresponding to the attribute $A_i$.)  
By using the order $<_\mathcal{A}$ over attributes, we can alternatively view $t$ as an \emph{unnamed} tuple, corresponding to the 
sequence $\bar t = t(A_1),\dots,t(A_n)$, with $A_1 <_\mathcal{A} \cdots <_\mathcal{A} A_n$. Thus, we can also view an instance $I$ as 
an assignment $R^I$ of sets of unnamed tuples (or just tuples) $\bar t \in D^n$.  
In general, when we know all attribute names for a relation, we use the unnamed perspective, but when the 
set of attributes is not clear, we resort to the named perspective. 
For the sake of readability, we abuse notation and use $\schema(I)$ to denote the schema of an instance $I$. 

For instances $I$ and $J$ over a schema $\Sch$, we write $I \subseteq J$ if for each relation symbol $R$ in $\Sch$ we have that 
$R^I \subseteq R^J$. 
Furthermore, if $I_1$ and $I_2$ are instances over respective schemas $\Sch_1$ and $\Sch_2$, 
we denote by $I_1 \cup I_2$ the instance over schema $\Sch_1 \cup \Sch_2$ 
such that $R^{I_1 \cup I_2} = R^{I_1} \cup R^{I_2}$ if $R$ is in both $\Sch_1$ and $\Sch_2$, 
$R^{I_1 \cup I_2} = R^{I_1}$ if $R$ is only in $\Sch_1$, and $R^{I_1 \cup I_2} = R^{I_2}$ if 
$R$ is only in $\Sch_2$. 
Finally, an instance $I'$ \emph{extends} an instance $I$ if (1) $\schema(I')$ extends $\schema(I)$, and (2)  
for each relation $R$ in $\schema(I)$ with assigned attributes $\{A_1,\dots,A_n\}$ and for each tuple $t$ in $R^I$, 
there is a tuple $t'$ in $R^{I'}$ such that $t(A_i) = t'(A_i)$ for each $1 \leq i \leq n$.  Intuitively, $I$ 
extends $I'$ if the projection of $I'$ over the schema of $I$ is contained in $I$.


\noindent
\textbf{Conjunctive queries}:  Since our goal is for queries to be applicable to different schemas, we adopt a named perspective on queries. 
A {\it named atom} is an expression of the form $R(A_1:x_1,\dots,A_k:x_k)$, where 
$R$ is a relation name, each $A_i$ is an attribute name, and each $x_i$ is a variable. We say that 
the variables mentioned by such an atom are $x_1,\dots,x_k$, and the attributes mentioned by it are 
$A_1,\dots,A_k$. A conjunctive query (CQ) is an expression of the form  
$\exists \bar z \phi(\bar z,\bar y)$, where $\bar z$ and $\bar y$ are tuples of variables 
and $\phi(\bar z,\bar y)$ is a conjunction of named atoms that use the variables in $\bar z$ and $\bar y$. 


A named atom $R(A_1:x_1,\dots,A_k:x_k)$ is \emph{compatible} with  
schema $\Sch$ if $\{A_1,\dots,A_k\} \subseteq \Sch(R)$. 
A CQ is compatible with $\Sch$ if all its named atoms are compatible. 
Given a named atom $R(A_1:x_1,\dots,A_k:x_k)$, an instance $I$ of a schema $\Sch$ 
that is compatible with the atom, and an assignment $\tau:
\{x_1,\dots,x_k\} \rightarrow D$ of values to variables,  
we say that $(I,\tau)$ satisfy $R(A_1:x_1,\dots,A_k:x_k)$ if there is a 
tuple $a:\mathcal A \rightarrow D$
matching values with $\tau$ on attributes in $R$ in the sense that
$a(A_i) = \tau(x_i)$ for each $1 \leq i \leq k$.  (Under the unnamed
perspective we would require a tuple $a$ in $R^I$ such that its
projection $\pi_{A_1,\dots,A_k} \bar a$ over attributes
$A_1,\dots,A_k$ is precisely the tuple $\tau(x_1),\dots,\tau(x_k)$.)
The usual semantics of conjunctive queries
now follows, extending the notion of assignments in the usual way.
Finally, given a conjunctive query $Q$ that is compatible with $\Sch$, the evaluation $Q(I)$ of $Q$ over $I$ 
is the set of all the tuples $\tau(x_1),\dots,\tau(x_k)$ such that $(I,\tau)$ satisfy $Q$. 

We also need to specify queries that extract all tuples stored 
in a given relation, regardless of the schema, as is done in SQL with the query \texttt{SELECT * FROM R}. 
To be able to do this, we also use what we call \emph{total queries},
which, as we do not need to know the arity of $R$,
are simply constructs of the form 
$R$, for a relation name $R$.  A total query of this form is compatible with a schema $\Sch$ if $\Sch(R)$ is defined, and the evaluation 
of this query over an instance $I$ over a compatible schema $\Sch$ is simply the set of tuples $R^I$. 


\noindent
\textbf{Data Constraints}: 
Most of our data constraints can be captured by \emph{tuple-generating dependencies (tgds)}, which are expressions of the form 
$\forall \bar x \big(\exists \bar y\phi(\bar x,\bar y) \rightarrow \exists \bar z \psi(\bar x,\bar z)\big)$, for conjunctive queries
$\exists \bar y\phi(\bar x,\bar y)$ and $\exists \bar y\psi(\bar
x,\bar z)$, and by \emph{equality-generating dependencies (egds)}, which are expressions of the 
form $\forall \bar x \big(\exists \bar y\phi(\bar x,\bar y) \rightarrow x = x'\big)$, for a conjunctive query
$\exists \bar y\phi(\bar x,\bar y)$ and variables $x,x'$ in $\bar x$.  
As usual, for readability we sometimes omit the universal quantifiers of tgds and egds. 
An instance $I$ satisfies a set $\Sigma$ of tgds and egds, written $I \models \Sigma$, if 
(1) the schema of $I$ is compatible with each conjunctive query 
in each dependency in $\Sigma$, and (2) every assignment 
$\tau: \bar x \cup \bar y \rightarrow D$ 
such that $(I,\tau) \models \phi(\bar x,\bar y) \rightarrow D$ 
can be extended into an assignment $\tau': \bar x \cup \bar y \cup \bar z \to D$ such that 
$(I,\tau') \models \psi(\bar x,\bar z)$.  


A tgd is \emph{full} if it does not use existentially 
quantified variables on the right-hand side. 
A set $\Sigma$ of tgds is {\em full} if each tgd in $\Sigma$ is full.
$\Sigma$ is {\em acyclic} if an acyclic graph is formed by
representing each relation mentioned in a tgd in $\Sigma$ as a node and
by adding an edge from node $R$ to $S$ if a tgd in $\Sigma$
mentions $R$ on the left-hand side and $S$ on the right-hand side.

\noindent
\textbf{Structure Constraints}:
Structure constraints are used to specify that schemas need to contain a certain 
relation or certain attributes. A structure constraint is a formula of the form 
$R[\bar s]$ or $R[*]$, where $R$ is a relation symbol,  
$\bar s$ is a tuple of attributes, and $*$ is a symbol not in 
$\mathcal A$ or $\mathcal R$ intended to function as a 
wildcard. 
A schema $\Sch$ satisfies a structure constraint $R[\bar s]$, denoted 
by $\Sch \models R[\bar s]$, if $\Sch(R)$ is defined, and each attribute in $\bar s$ belongs to $\Sch(R)$ 
The schema satisfies the constraint $R[*]$ if $\Sch(R)$ is defined.
For an instance $I$ over a schema $\Sch$ and a set $\Sigma$ of tgds, egds, and structure constraints, 
we write $(I, \Sch) \models \Sigma$ if $I$ satisfies each data constraint in $\Sigma$ and $\Sch$ satisfies 
each structure constraint in $\Sigma$.

\section{Procedures}
\label{ref-procedures}

In this section we formalize the notion of procedures that transform
data.  We view procedures as black boxes, and assume no knowledge
of or control over their inner workings.  Our reasoning about 
procedures is based on two properties: an input condition, or
\emph{precondition} on the state of the data that must hold for a procedure to be applicable, and an output condition, or
\emph{postcondition} on the state of the data that must hold 
after the application.  

\begin{example}
\label{exa-intro} 
Consider again the medical example discussed in the introduction, with a schema having two 
relations: \textit{LocVisits}, holding information about emergency-room visits in a geographical area, 
and \textit{EVisits}, holding visit information for an individual emergency room in a particular location. 
Suppose we know that a procedure is available that migrates the data from 
\textit{EVisits} to \textit{LocVisits}. We do not know how the procedure works, 
but we do know that once it has been applied, all tuples in \textit{EVisits} also appear in \textit{LocVisits}. 
In other words, this procedure can be described by the following information: 

\smallskip
\noindent
\textbf{Precondition:} The schema has relations \textit{LocVisits} and \textit{EVisits}, each with attributes 
\fid, \pid and \timestamp (standing for facility ID, patient insurance ID, and timestamp). 

\smallskip
\noindent
\textbf{Postcondition:} Every tuple from \textit{EVisits} is in \textit{LocVisits}. 

\smallskip
\noindent
\textbf{Scope and safety guarantees:} To rule out procedures that, for example, delete all the tuples from the database, we must be assured that our procedure only modifies the relation \textit{LocVisits}, and that 
it preserves all the tuples present in \textit{LocVisits} before the application of the procedure. We shall soon see how to encode these guarantees into 
our framework. 

Suppose that after a while, the requirements of one of the partner agencies of the organization impose an additional 
requirement: Relation \textit{LocVisits} should also contain information about the age of the patients. 
Suppose the organization also has a relation \textit{Patients}, where the patient age is recorded in attribute 
\textit{age}, together with \fid and \textit{patientId}. To migrate the patient ages into 
\textit{LocVisits}, one needs the following steps: First add the attribute \age to \textit{LocVisists}, and then 
update this table so that the patient ages are as recorded in \textit{Patients}. 
We observe that all the procedures involved in this operation can be captured using the same framework of 
preconditions, postconditions, and scope/safety guarantees that we used to capture the data-migration procedure. 
\end{example}

\subsection{Formal Definition}

We define procedures with respect to a class $\Const$ of constraints and a class $\Que$ of queries.
\begin{definition}
A 
{\em procedure} $P$ over $\Const$ and $\Que$ is a tuple $(\Scope,\Cpre,\Cpost,\Q_\safe)$, 
where 
\begin{itemize}
\item $\Scope$ is a set of structure constraints that defines the scope (i.e., relations and attributes) in which  
the procedure acts; 
\item $\Cpre$ and $\Cpost$ are constraints in 
$\Const$ that describe the pre- and postconditions of the procedure,
respectively; and 
\item $\Q_\safe$ is a set of queries in $\Que$ that serve as a safety guarantee for the procedure. 
\end{itemize}
\end{definition}

\begin{example} 
\label{exa-proc-1}
Let us return to the procedure outlined in Example \ref{exa-intro}, where the intention was to define migration of 
 data from relation \textit{EVisits} into \textit{LocVisits}. In our formalism, we describe this procedure as follows. 
 
 \smallskip
\noindent
\underline{$\Scope$}: Since the procedure migrates tuples into $\textit{LocVisits}$, the scope of the procedure 
is just this relation. This is described using the structure constraint $\textit{LocVisits}[*]$.

\noindent
\underline{$\C_\pre$}: We use the structure constraints $\textit{EVisits}[\fid, \pid, \timestamp]$ and $\textit{LocVisits}$ $[\fid,$ $\pid,$  $\timestamp]$, to ensure that the database has the correct attributes. 

\noindent
\underline{$\C_\post$}: The postcondition comprises the tgd
\vspace{-5pt}
$$\textit{EVisits}(\fid:x, \pid:y, \timestamp:z)\rightarrow 
\textit{LocVisits}(\fid:x, \pid:y, \timestamp:z).$$
 That is to say, after the procedure has been applied,  
the projection of \textit{EVisits} over \fid, \pid and \timestamp is a subset of the respective projection of \textit{LocVisits}.

\noindent
\underline{$\Q_\safe$}: We can add safety guarantees in terms of queries that need to be preserved 
when the procedure is applied. In this case, since we do not want the procedure to delete anything 
that was stored in \textit{LocVisits} before the migration, we add the safety constraint 
$\textit{LocVisits}(\fid:x,\pid:y,\timestamp:z)$, whose intent is to to state that all answers to this query on \textit{LocVisits} that are present in the database before the application of the procedure 
must be preserved. We formalize this intuition when giving the semantics below. 
\end{example}

\subsection{Semantics} 

Formalizing the semantics of procedures requires additional notation. 
Given a set $\C$ of structure constraints and a  schema $\Sch$, we denote by 
$Q_{\Sch \setminus \C}$ the conjunctive query that, intuitively, is meant to retrieve the projection of the entire database over all relations and attributes not mentioned in $\C$. 
Formally, $Q_{\Sch \setminus \C}$ includes a conjunct
$R(A_1:z_1,\dots,A_m:z_m)$ for each relation $R$ in $\Sch$ but not
mentioned in $\C$, where $\Sch(R) = \{ A_1,\dots,A_m \}$ and $z_1,\dots,z_m$ are fresh variables.
In addition, if some constraint in $\C$ mentions a relation $T$ in
$\Sch$, but no constraint in $\C$ is of the form $T[*]$, then
$Q_{\Sch \setminus \C}$ also includes a conjunct
$T(B_1:z_1,\dots,B_k:z_k)$, where $\{B_1,\dots,B_k\}$ is the set of
all the attributes in $\Sch(T)$ that are not mentioned in any constraint in $\C$, 
and $z_1,\dots,z_k$ are again fresh variables.
For example, consider a schema $\Sch$ with relations $R$, $S$, and $T$, where 
 $R$ has attributes $A_1$ and $A_2$, $T$ has attributes $B_1$, $B_2$ and $B_3$, and 
 $S$ has $A_1$ and $B_1$. Further, consider the set $\C$ with a single constraint $R[*] \wedge S[B_1]$. 
 Then  $Q_{\Sch \setminus C}$ is the query $T(B_1: z_1, B_2:  z_2, B_3: z_3) \wedge S(A_1: w_1)$. 
 Note that $Q_{\Sch \setminus \C}$ is unique up to the renaming of
 variables and order of conjuncts. 
 
A procedure $P = (\Scope,\Cpre,\Cpost,\Q_\safe)$ is \emph{applicable} on an instance $I$ over schema $\Sch$ if 
(1) The query  $Q_{\Sch \setminus \Scope}$ and each query in $\Q_\safe$ are compatible with both $\Sch$ and $\Sch'$, and 
(2) $(I,\Sch)$ satisfy the preconditions $\C_\pre$. 
We can now proceed with the semantics of procedures. 
\begin{definition}
Let $I$ be an instance over a schema $\Sch$. 
An instance $I'$ over schema $\Sch'$ 
is a {\em possible outcome} of applying $P$ over the instance and schema $(I,\Sch)$ if the following holds: 
\begin{enumerate}
\item $P$ is applicable on $I$. 
\item $(I',\Sch') \models \C_\post$. 
\item The answers of the query $Q_{\Sch \setminus \Scope}$ do not change: 
$Q_{\Sch \setminus \Scope}(I) = Q_{\Sch \setminus \Scope}(I')$.
\item The answers of each query $Q$ in $\Q_\safe$ over $I$ are preserved: 
$Q(I) \subseteq Q(I')$. 
\end{enumerate} 
\end{definition}

In the definition, we state the schemas of instances $I$ and $I'$ explicitly, to reinforce the 
fact that schemas may change during the application of procedures. 
However, most of the time the schema can be understood from the instance, so we normally 
just say that an instance $I'$ is a possible outcome of $I$ (even if the schemas of $I$ and $I'$ are different). Let us also recall that we use $\schema(I)$ to denote the 
schema of an instance $I$.

\begin{figure}[t]
\centering
{\small
\begin{tabular}{c}
\fbox{\parbox{0.93\textwidth}{
\centering
\begin{tabular}{ccc}
\evisits &  & \locvisits \\ 
\begin{tabular}{|ccc|}
\hline
\fid & \pid & \timestamp \\
\hline 
1234 & 33 & {070916 12:00} \\
2087 & 91 & {090916 03:10} \\
\hline
\end{tabular}
& & 
\begin{tabular}{|ccc|}
\hline
\fid & \pid & \timestamp \\
\hline
1234 & 33 & {070916 12:00} \\
1222 & 33 & {020715 07:50} \\
\hline
\end{tabular}
\end{tabular} \\ 
(a) Instance $I$}}\\
\\
\fbox{\parbox{0.93\textwidth}{
\centering
\begin{tabular}{ccc}
\evisits &  & \locvisits \\ 
\begin{tabular}{|ccc|}
\hline
\fid & \pid & \timestamp \\
\hline 
1234 & 33 & {070916 12:00} \\
2087 & 91 & {090916 03:10} \\
\hline
\end{tabular}
& & 
\begin{tabular}{|ccc|}
\hline
\fid & \pid & \timestamp \\
\hline
1234 & 33 & {070916 12:00} \\
1222 & 33 & {020715 07:50} \\
2087 & 91 & {090916 03:10} \\
\hline
\end{tabular}
\end{tabular} \\
(b) Possible outcome $J_1$ of applying $P$ over $I$}}
\end{tabular}
\smallskip

\begin{tabular}{cc}
\fbox{\parbox{0.42\textwidth}{
\centering
\begin{tabular}{c}
\locvisits \\ 
\begin{tabular}{|ccc|}
\hline
\fid & \pid & \timestamp \\
\hline 
1234 & 33 & {070916 12:00} \\
1222 & 33 & {020715 07:50} \\
2087 & 91 & {090916 03:10} \\
4561 & 54 & {080916 23:45} \\
\hline
\end{tabular} \\ 
(c) relation \locvisits in $J_2$
\end{tabular}}} & 
\fbox{\parbox{0.50\textwidth}{
\centering
\begin{tabular}{c}
\locvisits \\ 
\begin{tabular}{|cccc|}
\hline
\fid & \pid & \timestamp & \age \\
\hline 
1234 & 33 & {070916 12:00} & 21\\
1222 & 33 & {020715 07:50} & 45 \\
2087 & 91 & {090916 03:10} & 82 \\
\hline
\end{tabular} \\ 
(d) relation \locvisits in $J_3$
\end{tabular}}} 
\end{tabular}
}
\vspace{-5pt}
\caption{Instance $I$ of Example \ref{exa-proc-1b} (a), a complete possible outcome (b), and 
the relation \locvisits of two other possible outcomes, one in which \locvisits contains additional tuples not mentioned in \evisits (c), and one where an extra attribute is added to \locvisits (d).}
\label{fig-exa-1}
\end{figure}

\begin{example}[Example \ref{exa-proc-1} continued]
\label{exa-proc-1b}
Recall the procedure $P = (\Scope,\Cpre,\Cpost,\Q_\safe)$  defined in Example \ref{exa-proc-1}. 
Consider the instance $I$ over the schema $\Sch$ with relations \evisits and \locvisits, each 
with attributes \fid, \pid, and \timestamp, as shown in Figure \ref{fig-exa-1} (a). 
Note first that $P$ is indeed applicable on $I$. 
When applying the procedure $P$ over $I$, we know from $\Scope$ that the only relation whose content can change is 
\locvisits, while \evisits (or more precisely, the projection of \evisits over \pid, \fid and \timestamp) is the same across all possible outcomes. 
Furthermore, we know from $\Cpost$ that in all possible outcomes the projection of 
\evisits over attributes \fid, \pid, and \timestamp must be the same as the projection of 
\locvisits over the same attributes. Finally, from $\Q_\safe$ we know that 
the projection of \locvisits over these three attributes must be preserved. 

Perhaps the most obvious possible outcome of applying $P$ over $I$ is that of the instance $J_1$ in Figure \ref{fig-exa-1} (b), 
corresponding to the outcome where the tuple in \evisits that is not 
yet in \locvisits is migrated into this last relation. However, since we assume no control over the actions performed by the procedure $P$, 
it may well be that it is also migrating data from a different relation that we are not aware of, 
producing an outcome whose relation \evisits remains the same as in $I$ and $J_1$, but 
where \locvisits has additional tuples, as depicted in Figure \ref{fig-exa-1} (c). Moreover,
it may also be the case that the procedure alters the schema of \locvisits, adding an extra attribute \age, importing the information from an unknown source, as shown in Figure \ref{fig-exa-1} (d). 
\end{example}  

As we have seen in this example, in general the number of possible outcomes (and even the number of possible schemas) that result after a procedure is executed 
is infinite. For this reason, we are generally more interested in properties shared by all possible outcomes, which motivates the following definitions. 

\begin{definition}
The {\em outcome set} of applying a procedure $P$ to $I$ is defined as the set. 
$$\outcome_P(I) = \{I' \mid I' \text{ is a possible outcome of applying } P \text{ to } I \}.\footnote{Recall that the schema of instances $I'$ 
is not necessarily the same as that of $I$.}$$
\end{definition} 


\section{Defining Common Database Tasks as Procedures}
\label{sec-exa}

We now show additional examples of defining common database tasks  as procedures within our framework. 
We show that data exchange, alter-table statements, and data cleaning can all be accommodated by the framework, 
and provide additional examples in Appendix \ref{app-exa}. It is worth noticing that in our first three examples we use only structure 
constraints, tgds, and egds as pre- and postconditions, and that our
safe queries are all conjunctive queries.  The last example calls for extending 
the language used to define procedures. 

\subsection{Data Exchange} 

\newcommand{\Scht}{\Sch_\textit{t}} 
\newcommand{\Schs}{\Sch_\textit{s}} 

We have already seen an example of specifying data-migration tasks as black-box procedures. 
However, a more detailed  discussion will allow us to illustrate some of the basic properties of our framework. 
Following the notation introduced by Fagin et al. in \cite{FKMP05}, the most basic instance of the data-exchange problem 
considers a source schema $\Sch_\text{s}$, a target schema $\Sch_\text{t}$,  
and a set $\Sigma$ of dependencies that define how data from the source schema are to be mapped to the target schema. 
The dependencies in $\Sigma$ are usually tgds whose left-hand side is compatible 
with $\Sch_\text{s}$, and the right-hand side is compatible with 
$\Scht$. 
The data-exchange problem is as follows: Given a source instance $I$, compute 
a target instance $J$ so that $(I,J)$ satisfies all the dependencies in $\Sigma$.  
Instances $J$ with this property are called \emph{solutions} for $I$ under $\Sigma$. 
 
In encapsulating this task as a black box within our framework, we assume that the target and source  schemas are part of the same schema.
(Alternatively, one can define procedures working over different databases.) 
Let $(\Sch_\text{s},\Sch_\text{t},\Sigma)$ be as above. 
We construct the procedure $\Pro^\textit{st} = (\Scope^\textit{st},\Cpre^\textit{st},\Cpost^\textit{st},\Q_\safe^\textit{st})$,  
where 
\begin{itemize}
\item $\Scope^\textit{st}$ contains an atom $R[*]$ for each relation $R$ on the right-hand side of a tgd in $\Sigma$; 
\item $\Cpre^\textit{st}$ contains a structure constraint 
$Q[A_1,\dots,A_n]$
for each query of the form\linebreak $Q(A_1:x_1,\dots,A_n:x_n)$ on the left-hand side of a tuple-generating dependency in $\Sigma$; 
\item $\C_\post^\textit{st}$ is the set of all the tgds in $\Sigma$; and 
\item $\Q_\safe^\textit{st}$ is the conjunction of all the atoms $R$ appearing in any tgd in $\Sigma$. 
\end{itemize}
By the semantics of procedures, it is not difficult to conclude that, 
for every pair of instances $I$ and $J$ over $\Sch_\text{s}$ and 
$\Sch_\text{t}$, respectively, we have that 
$J$ is a solution for $I$ if and only if the instance $I \cup J$ over 
the schema $\Sch_\text{s} \cup \Sch_\text{t}$ is a possible outcome of applying 
$\Pro^\textit{st}$ 
over $(I,\Sch_\text{s} \cup \Sch_\text{t})$. We can make this statement much more general, 
as the set of all possible outcomes essentially corresponds to the set of solutions of the data-exchange setting. 

\begin{proposition}
An instance $J$ over schema $\Sch_\text{s} \cup \Sch_\text{t}$ is a possible outcome of applying $\Pro^\textit{st}$ over 
$(I,\Sch_\text{s} \cup \Sch_\text{t})$ if and only if $J$ is a solution for $I$ under $\Sigma$. 
\end{proposition}



\subsection{Alter Table Statements}

In our framework, procedures can be defined to work over more than one schema, 
as long as the schemas satisfy the necessary input and compatibility conditions. This is inspired by SQL, 
where statements such as \texttt{INSERT INTO R (SELECT * FROM S)} would be executable over any schema,  
as long as the relations $R$ and $S$ have the same types of attributes in the same order. 
Thus, it seems logical to allow procedures that alter the schema of the 
existing database. To do so, we use structure constraints, as shown in the following example. 

\begin{example}
\label{exa-alter}
Recall from Example \ref{exa-intro} that,  
due to a change in the requirements, we now need to add the attribute \age to the schema of \locvisits. 
In general, we capture \emph{alter table} statements by procedures without scope, used only 
to alter the schema of the outcomes, so that it would satisfy the structural postconditions of procedures. 
In this case, we model a procedure that adds \age to the schema of \locvisits with the 
procedure $P' = (\Scope',\Cpre',\Cpost',\Q_\safe')$, where $\Scope'$ and $\Q_\safe'$ are empty 
(if there is no scope, then the database does not change modulo adding attributes, so we do not 
include any safety guarantees), $\Cpre'$ is the stucture constraint $\locvisits[*]$, stating that the relation 
exists in the schema, and $\Cpost'$ is the structure constraint $\locvisits[\age]$, stating that 
\locvisits now has an age attribute. 
Note that the instance $J_3$ in Figure \ref{fig-exa-1}(d) 
with \evisits as in $J_1$ in Figure \ref{fig-exa-1}(b), is actually a possible outcome of applying $P'$ over instance $J_1$;
the part of the instance given by the schema of $J_1$ does not change, but we do add an extra attribute 
\age to \locvisits, and we cannot really control the values of the newly added attribute. 
\end{example}
We remark that the empty scope in $P'$ guarantees that no relations or attributes are 
deleted when applying this procedure. This happens because 
$Q_{\Sch \setminus \Scope}$ must be compatible with the schema of all outcomes. 
However, nothing prevents us from adding extra attributes on top of \age. This decision  
to use the open-world assumption on schemas reflects the
understanding of procedures as black boxes, which we can execute but
not control in other ways.

\subsection{Data Cleaning} 

Data cleaning is a frequent and important task within database systems (see e.g., \cite{2012FanGeertsBook}). The most simple 
cleaning scenario one could envision is when we have a relation $R$ whose attribute values   
are deemed incorrect or incomplete, and it is desirable to provide the correct values. There are, in general, multiple 
ways to do this; here we consider just a few of them. 

The first possibility is to assume that we have the correct values in another relation, and 
to use this other relation to provide the correct values for $R$. Consider an example. 

\begin{example}
Consider again the schema from Example \ref{exa-intro}. Recall that in Example \ref{exa-alter} we added the attribute \age to the schema of \locvisits. The problem is that we have no control over the newly added values of \age. (If the procedure was a SQL 
alter-table statement, then the \age column would be filled with nulls.) However, another relation, 
\patients, associates an \age value with each pair of (\fid, \pid) values; all we need to do now is to copy the appropriate 
\age value into each tuple in \locvisits. To this end, we specify the procedure $P^* = (\Scope^*,\Cpre^*,\Cpost^*,\Q_\safe^*)$,  
which copies the values of $\age$ from $\patients$ into 
$\locvisits$, using the values of $\fid$ and $\pid$ as a reference. 

\noindent
$\Scope^*$: We use the constraint $\locvisits[\age]$, so that the only 
piece of the database the procedure can alter is $\age$ in the relation $\locvisits$.

\noindent
$\C_\pre^*$: The preconditions are the structure constraints $\locvisits[\fid,\pid,\age]$ and 
$\patients[\fid,\pid,\age]$, plus the fact that the values of $\fid$ and $\pid$ need to determine the values of 
$\age$ in the \patients relation, specified with the dependency  
$\patients(\fid:x,\pid:y, \age:z) \wedge \patients(\fid:x,\pid:y, \age:w) \rightarrow z = w$. 
Note that in this case we do not actually need the structure constraints in \patients, because they are implicit in the dependencies (they need to be compatible with the schema), but we keep them for clarity. 

\noindent
$\C_\post^*$: The postcondition is the constraint
$\patients(\fid:x,\pid:y,\age:z) \wedge \locvisits(\fid:x,\pid:y,\age:w) \rightarrow z = w.$
Alternatively, if we know that 
all the $(\fid,\pid)$ pairs from \patients are in \locvisits (which can be required with a precondition), 
we can specify the same postcondition via  
$\locvisits(\fid:x,\pid:y, \age:w) \rightarrow 
\patients(\fid:x,\pid:y, \age:w)$.

\noindent
$\Q_\safe^*$:  Same as before, no guarantees are needed. 

As desired, in all the outcomes of $P^*$ the value of the \age attribute in \locvisits is the same as in the corresponding 
tuple (if it exists) in \patients with the same \fid and \pid values. But then again, the procedure might modify the schema of some relations, or might even create auxiliary relations in the database in the process. 
What we gain is that this procedure will work regardless of the shape of relations \locvisits and \patients, as long 
as the schemas satisfy the compatibility and structure constraints. 
\end{example}
In the above example we used a known auxiliary relation to clean the values of \age in \locvisits. Alternatively, we could define a  more 
general procedure that would, for instance, only remove nulls from \locvisits, without controlling which values end up replacing these nulls. 
In order to state this procedure, let us augment the language of tgds with an auxiliary predicate $C$ (for \emph{constant}) with a single attribute \textit{val}, which is to take the role of the \texttt{NOT NULL} constraint in SQL: It is true only for the non-null values in $D$. 

\begin{example}
Let us now define a procedure $\hat P = (\hat \Scope, \hat \Cpre, \hat \Cpost,  \hat \Q_\safe)$ that simply replaces all null values of the attribute \age in relation \locvisits with non-null values. 

\noindent
$\hat \Scope$: The scope is again $\locvisits[\age]$, just as in the previous example.

\noindent
$\hat \C_\pre$: In contrast with the procedure $P^*$ of the previous example, this procedure is light on preconditions: We only need 
relation \locvisits to be present and have the \age attribute. 

\noindent
$\hat \C_\post$: The postcondition states that the attribute \age of \locvisits no longer has null values. To express this, we use the auxiliary predicate $C$, and 
define the constraint $\locvisits(\age: x) \rightarrow C(\val : x )$, which states that no value in the attribute \age in \locvisits is null. 

\noindent
$\hat \Q_\safe$:  Since we only want to eliminate null values, we also include the safety query $\locvisits(age : x, \fid: y, \pid: z) \wedge C(\val : x)$, so that we preserve all the 
non-null values of \age (with the correct \fid and \pid attached to these ages). 
\end{example}

\section{Basic Computational Tasks for Relational Procedures}
\label{ref-basics}

In this section we study some formal properties of our procedure-centric framework, with the intent of showing how the proposed  framework 
can be used as a toolbox for reasoning about sequences of database procedures.  
We focus on what we call \emph{relational procedures}, where the sets of pre- and postconditions are given 
by tgds, egds, or structure constraints, and safety queries can be conjunctive or total queries. 
While there clearly are interesting classes of procedures that do not fit into this special case in the proposed framework, we remark that relational procedures 
are general enough to account for a wide range of relational operations on data, including the examples in the previous section.





\subsection{Applicability} 

In the proposed  framework we focus on transformations of data sets given by sequences of procedures. 
Because we treat procedures as black boxes, the only description we have of the 
results of these transformations is that they ought to satisfy the output constraints of the procedures. 
In this situation, how can one guarantee that all the procedures will be applicable? 
Suppose that, for instance, we wish to apply procedures $P_1$ and $P_2$ 
to an instance $I$ in sequential order: First $P_1$, then $P_2$. 
The problem is that, since 
output constraints do not fully determine the outcome of $I$ after applying $\Pro_1$, we cannot 
immediately guarantee that this outcome 
is an instance that satisfies 
the preconditions of $P_2$. 

Given that the set of possible outcomes is in general infinite, our focus is on  
guaranteeing that \emph{any} possible outcome of applying $P_1$ over $I$ will satisfy the preconditions 
of $P_2$. To formalize this intuition, we need to extend the notion of outcome to a set of instances. 
We define the outcome of applying 
a procedure $P$ to a set of instances $\Inst$ as 
$$\outcome_P(\Inst) = \bigcup_{I \in \Inst} \outcome_P(I),$$ 
the union of the outcomes of all the instances in $\Inst$. 
Furthermore, for a sequence $\Pro_1,\dots,\Pro_n$ of procedures we define the outcome of applying $\Pro_1,\dots,\Pro_n$ to an instance 
$I$ as the set 
$$\outcome_{\Pro_1,\dots,\Pro_n}(I) = \outcome_{P_n}(\outcome_{P_{n-1}}(\cdots (\outcome_{P_1}(I)) \cdots )).$$

We can now define the first problem of interest: 
\vspace{-12pt}
\begin{center}
\fbox{
\begin{tabular}{ll}
\multicolumn{2}{l}{\textsc{Applicability:}}\\
\vspace{-12pt}
& \\
\textbf{Input}: & A sequence $\Pro_1,\dots,\Pro_n$ of procedures, a schema $\Sch$; \\
\textbf{Question}: & Is it true that, for any arbitrary instance $I$ over $\Sch$, procedure $P_n$ can \\ 
& be applied  to each instance in the set  $\outcome_{\Pro_1,\dots,\Pro_{n-1}}(I)$?
\end{tabular}
}
\end{center}
It is not difficult to show that the \textsc{Applicability} problem is intimately related to the 
problem of implication of dependencies, defined as follows: Given a set $\Sigma$ of dependencies  
and an additional dependency $\lambda$, is it true that all the instances that satisfy $\Sigma$ also satisfy $\lambda$ ---  
that is, does $\Sigma$ imply $\lambda$?  
Indeed, consider a class $\mathcal L$ of constraints for which the implication problem is known 
to be undecidable. 
Then one can easily show that the applicability problem is also 
undecidable for those procedures whose pre- and postconditions 
are in $\mathcal L$: Intuitively, if we let $\Pro_1$ be a procedure with a set $\Sigma$ 
of postconditions, and $\Pro_2$ a procedure with a dependency $\lambda$ as a precondition, then it is not difficult 
to come up with proper scopes and safety queries so that 
$\outcome_{\Pro_1}(I)$ satisfies $\lambda$ for every instance $I$ over schema $\Sch$ if and only if 
$\lambda$ is true in all instances that satisfy $\Sigma$. However, as the following proposition shows, the applicability problem 
is undecidable already for very simple procedures, and even when we consider the \emph{data-complexity} 
view of the problem, that is when we fix the procedure and take a particular input instance. 

\begin{proposition}
\label{prop-ap-und}
There are fixed procedures $\Pro_1$ and $\Pro_2$ that only use tgds for their constraints, and such that the following problem 
is undecidable. Given an instance $I$ over schema $\Sch$, is it true that all the instances in 
$\outcome_{\Pro_1}(I)$ satisfy the preconditions of $\Pro_2$? 
\end{proposition}

The proof of Proposition \ref{prop-ap-und} is by reduction from the embedding problem for finite semigroups, 
shown to be undecidable in \cite{KPT06}. 

There are several lines of work aiming to identify practical classes of constraints for which the implication problem 
is decidable, and all that work can be applied in our framework. 
However, we  opt for a stronger restriction: Since all of our examples so far use only structure constraints 
as preconditions, for the remainder of the paper we focus on procedures whose preconditions comprise 
structure constraints. In this setting, we have the following result. 

\begin{proposition}
\label{ref-rep-decidable}
\textsc{Applicability} is in polynomial time for sequences of relational procedures whose preconditions contain only structure constraints.
\end{proposition}

\subsection{Representing the Outcome Set}

We have seen that deciding properties about the outcome set of a sequence of procedures (or even of a single procedure)  
can be a complicated task. One of the reasons is that procedures do not completely define their outcomes: We 
do not really know what will be the outcome of applying a sequence $\Pro_1,\dots,\Pro_n$ of procedures to an 
instance $I$, we just know it will be an instance from the collection $\outcome_{\Pro_1,\dots,\Pro_n}(I)$. 
This collection may well be of infinite size, but can it still be represented finitely? The database-theory community 
has developed multiple formalisms for representing sets of database instances, from notions of tables with incomplete information \cite{IL84} to 
knowledge bases (see, e.g., \cite{BO15}). In this section we study the possibility of 
representing outcomes of (sequences of) procedures by means of incomplete tables, along the lines of \cite{IL84}. We also discuss some negative results 
about representing outcomes of general procedures in systems such as knowledge bases, but leave a more 
detailed study in this respect for future work. 

The first observation we make is that allowing arbitrary tgds in procedures introduces 
problems with management of sequences of procedures. Essentially, 
any means of representing the outcome of a sequence of procedures 
needs to be so powerful that even deciding whether it is nonempty is going to be undecidable. 

\begin{proposition}
\label{prop-rep-undec}
There is a fixed procedure $\Pro$ that does not use preconditions and 
only use tgds in their postconditions, such that the following problem is 
undecidable: Given an instance $I$, is the set 
$\outcome_{\Pro_1}(I)$ nonempty? 
\end{proposition}

The reason we view Proposition \ref{prop-rep-undec} as a negative result is because it rules out the 
possibility of using any ``reasonable'' representation system. Indeed, one would expect that deciding non-emptiness 
should be decidable in any reasonable way of representing infinite sets of instances. 
Proposition \ref{prop-rep-undec} is probably not surprising, since reasoning about tgds in general is 
known to be a hard problem. Perhaps more interestingly, in our case one can show that the above fact remains true even if 
one allows only \emph{acyclic} tgds, which are arguably one of the most well-behaved classes of dependencies in the literature. 
The idea behind the proof is that one can simulate cyclic tgds via procedures with only acyclic tgds and no scope. 

\begin{example}
\label{exa-scope}
Consider two procedures $\Pro_1$ and $\Pro_2$, where $P_1 = (\Scope^1, \Cpre^1,\Cpost^1, \Q_\safe^1)$, with 
$\Scope^1 = \{R[*],T[*]\}$, $\Cpre^1 = \emptyset$, $\Cpost^1 = \{R(A:x) \rightarrow T(A:x)\}$ and  $\Q_\safe^1 = R(A:x) \wedge T(A:x)$;  
$\Pro_2$ has empty scope, preconditions, and safety queries, and has postconditions $\{T(A:x) \rightarrow R(A:x) \}$. 
Let $I$ be an instance over the schema with relations $R$ and $T$, both with attribute $A$. By definition, 
the set of possible outcomes of $\Pro_1$ over $I$ are all instances $J$ that extend $I$ and satisfy 
the dependency $R(A:x) \rightarrow T(A:x)$. However, the set $\outcome_{P_1,P_2}(I)$ corresponds to all instances 
$I'$ that extend $I$ and satisfy both 
dependencies $R(A:x) \rightarrow T(A:x)$ and $T(A:x) \rightarrow R(A:x)$ (In other words, 
we can use $\Pro_2$  to \emph{filter out} all those instances $J$ where $T^J  \not\subseteq R^J$). Intuitively, this happens 
because the outcome set of applying $\Pro_2$ over any instance not satisfying $T(A:x) \rightarrow R(A:x)$ is empty, and 
we define $\outcome_{P_1,P_2}(I)$ as the union of each set $\outcome_{P_2}(K)$, for each instance $K \in \outcome_{P_1}(I)$. 
\end{example} 
By applying the idea of this example to the proof of Proposition \ref{prop-rep-undec}, we show: 
\begin{proposition}
Proposition \ref{prop-rep-undec} holds for procedures $\Pro_1$ and $\Pro_2$ that only use acyclic tgds. 
 \end{proposition}

Since acyclic tgds do not help, we may consider restrictions to full tgds. 
Still, even this is not enough for making the non-emptiness problem decidable, once one adds the possibility of having schema constraints in procedures. 

\begin{proposition}
\label{prop-rep-undec-nothingworks}
There exists a sequence $\Pro_1,\Pro_2,\Pro_3$ of procedures such that the following problem is 
undecidable:  Given an instance $I$, is the set $\outcome_{\Pro_1,\Pro_2,\Pro_3}(I)$ nonempty? 
Here, all the procedures have no preconditions, and have postconditions built using acyclic sets of full tgds and schema constraints (and nothing else). 
\end{proposition}


Propositions \ref{prop-rep-undec} and \ref{prop-rep-undec-nothingworks} tell us that restricting the classes of dependencies allowed in procedures 
may not be enough to guarantee outcomes that can be represented by reasonable systems. Thus, we now adapt a different strategy: 
We restrict interplay between the postconditions  of procedures, their scope, and their safety queries. Let us define two important classes of procedures that will be used thoroughout this section. 

We say that procedure $P = (\Scope,\Cpre,\Cpost,\Q_\safe)$ is \emph{safe scope} if the following 
holds: 
\begin{itemize}
\item $\Cpost$ is a set of tgds where no relation in the right-hand side of a tgd appears also in the left-hand side of a tgd; 
\item The set $\Scope$ contains exactly one constraint $R[*]$ for each relation $R$ that appears 
on the right-hand side of a tgd in $\Cpost$; and 
\item The query $Q_\safe$ corresponds to $\bigwedge_{R[*] \in \Scope} R$,  that is it binds precisely all the 
relations in the scope of $P$. 
\end{itemize}
(For instance, procedure $P$ in Example 
\ref{exa-proc-1} is essentially a procedure with safe scope, as it can 
easily be transformed into one by slightly altering the safety query.) 

We also define a class of procedures that ensure that certain attributes or relations  
be present in the schema. Formally, we say that a 
procedure $P = (\Scope,\Cpre,\Cpost,\Q_\safe)$ is an \emph{alter-schema procedure} if the following holds: 
\begin{itemize}
\item Both $\Scope$ and $\Q_\safe$ are empty; and 
\item $\Cpost$ is a set of structure constraints. 
\end{itemize} 

Let $\P^{\textit{safe},\textit{alter}}$ be the class of all the procedures that are either safe scope or alter-schema procedures. 
The class $\P^{\textit{safe},\textit{alter}}$ allows for practically-oriented interplay between migration and schema-alteration tasks  
and, as we will see in this section, is more manageable from the point of view of reasoning tasks, in terms of complexity. 
To begin with, deciding the non-emptiness of a sequence of procedures is essentially tractable for $\P^{\textit{safe},\textit{alter}}$: 

\begin{theorem}
\label{theo-rep-safe-scope}
The problem of deciding, given an instance $I$ and a sequence $\Pro_1,\dots,\Pro_n$ of procedures in $\P^{\textit{safe},\textit{alter}}$, 
whether $\outcome_{\Pro_1,\dots,\Pro_n}(I) \neq \emptyset$, is in exponential time, and is polynomial if the number $n$ of procedures is fixed. 
\end{theorem}

The proof of Theorem \ref{theo-rep-safe-scope} is based on the idea of chasing instances with the dependencies in the procedures, and of adding 
attributes to schemas as dictated by the alter-schema procedures. As usual, to enable the chase  
we need to introduce labeled nulls in instances (see, e.g., \cite{IL84,FKMP05}), and composing procedures calls for extending the techniques 
of \cite{APR13} to enable chase instances that already have null values. Using  the enhanced approach, 
one can show that the result of the chase is a good over-approximation of the outcome of a sequence of procedures. To state this result, 
we introduce conditional tables \cite{IL84}. 
 
Let $\Nulls$ be an infinite set of   {\em null values} that is disjoint from the set of domain values $D$. 
A {\em naive instance} $T$ over schema $\Sch$ assigns a finite relation $R^T \subseteq (D \cup \Nulls)^n$ 
to each relation symbol $R$ in $\Sch$ of arity $n$. 
Conditional instances extend naive instances by attaching conditions over the tuples. 
Formally, an \emph{element-condition} is a positive boolean combination of formulas 
of the form $x = y$ and $x \neq y$, where $x \in \Nulls$ and $y \in (D \cup \Nulls)$. 
Then, a {\em conditional instance} $T$ over schema $\Sch$ assigns to each $n$-ary relation 
symbol $R$ in $\Sch$  a pair $(R^T,\rho^T_R)$, where $R^T \subseteq (D \cup \Nulls)^n$ and 
$\rho^T_R$ assigns an element-condition to each tuple $t \in R^T$. 
A conditional instance $T$ is \emph{positive} if none of the element-conditions in its tuples uses 
inequalities (of the form $x \neq y$). 



To define the semantics, let $\nulls(T)$ be the set of all nulls in any tuple in $T$ or in an element-condition used in $T$. 
Given a substitution $\nu: \nulls(T) \rightarrow D$, let $\nu^*$ be the 
extension of $\nu$ to a substitution $D \cup \nulls(T) \rightarrow D$ that is the identity on $D$. We say that 
$\nu$ {\em satisfies an element-condition} $\psi$, and write $\nu \models \psi$, if for every equality $x= y$ in $\psi$ it is the case that 
$\nu^*(x) = \nu^*(y)$ and for every inequality $x \neq y$ we have that $\nu^*(x) \neq \nu^*(y)$.  
Furthermore, we define the set $\nu(R^T)$ as $\{\nu^*(t) \mid t \in R^T$ and $\nu \models \rho^T_R(t)) 
\}$. 
Finally, for a conditional instance $T$, $\nu(T)$ is the 
instance that assigns $\nu(R^T)$ to each relation $R$ in the schema. 

The set of instances represented by $T$, denoted by $\rep(T)$, is defined as 
$\rep(T) = \{I \mid$ there is a substitution $\nu$ such that $I$ extends $\nu(T) \}$. 
Note that the instances $I$ in this definition could have potentially bigger schemas than $\nu(T)$, or, in other words, 
we consider the set $\rep(T)$ to contain instances over any schema extending the schema of $T$. 

The next result states that conditional instances are good over-approximations for the outcomes 
of sequences of procedures. More interestingly, these approximations preserve the \emph{minimal instances} 
of outcomes. To put this formally, we say that an instance $J$ in a set $\Inst$ of instances is {\em minimal} if there is no instance $J' \in \Inst, J' \neq J$, and  
such that $J$ extends $J'$. 

\begin{proposition}
\label{prop-minimal}
Let $I$ be an instance and $\Pro_1,\dots,\Pro_n$ be a sequence of procedures in $\P^{\textit{safe},\textit{alter}}$. 
Then either $\outcome_{\Pro_1,\dots,\Pro_n} = \emptyset$ or one can construct, in exponential time (or polynomial if $n$ is fixed), a  
conditional instance $T$ 
such that 
\begin{itemize}
\item $\outcome_{\Pro_1,\dots,\Pro_n}(I) \subseteq \rep(T)$; and 
\item If $J$ is a minimal instance in $\rep(T)$, then $J$ is also minimal in 
$\outcome_{\Pro_1,\dots,\Pro_n}(I)$. 
\end{itemize}
\end{proposition}

We remark that this proposition can be extended to include procedures defined only with egds, at the cost of a much more technical presentation. 
While having an approximation with these properties is useful for reasoning tasks related to CQ answering, or 
in general checking any criterion that is closed under extensions of instances,  there is still the question of 
whether one can find any reasonable class of properties whose entire outcomes can be represented by these tables. 
However, as the following example shows, this does not appears to be possible, unless one is restricted to sequences 
of procedures almost without interaction with each other (see an example in appendix \ref{app-more-proofs}). 

\begin{example}
\label{exa-scope-2}
Consider a procedure $\Pro = (\Scope, \Cpre,\Cpost,\Q_\safe)$ with safe scope, where 
$\Scope = S[*]$, $\Cpre$ is empty, $\Cpost = \{R(A: x) \rightarrow S(A:x)\}$ and $\Q_\safe = S$. 
Consider now the conditional instance $T$ over the schema with relations $R$ and $S$, both with attribute $A$, given by 
$R^T= \{1,2\}$ and $S^T = \{1,2\}$. One could be tempted to say that $T$ is itself a representation of the set 
$\outcome_{\Pro}(\rep(T))$, and indeed $\rep(T)$ and $\outcome_{\Pro}(\rep(T))$ share their only minimal instance (essentially, 
the instance given by $T$). However, the open-world assumption behind $\rep(T)$ allows for instances that do not satisfy $\Cpost$, 
whereas all outcomes in $\outcome_{\Pro}(\rep(T))$ must satisfy $\Cpost$. One can in fact generalize this argument to show that 
conditional instances are not enough to fully represent outcome sets. 
\end{example}
Example \ref{exa-scope-2} suggest that one could perhaps combine conditional instances with a knowledge base, to allow for 
a complete representation of the outcome set of sequences of safe procedures. However, this would require studying the interplay of 
these two different types of representation systems, a line of work which is interesting in its own right.



\section{Future Work and Opportunities}
\label{ref-conc}

In this paper, we introduced basic building blocks for a proposed framework for assessing achievability of data-quality constraints. We  demonstrated that 
the framework is general enough to represent 
nontrivial database tasks, and exhibited realistic classes of procedures 
for which reasoning tasks can be tractable. 
Our next step is to address the problem of assessing achievability of constraints, which can be formalized as follows. 
Let $Q$ be a boolean query, $\Pi$ a set of procedures, and $I$ an instance over a schema $\Sch$. 
Then we say that $I$ can be readied for $Q$ using $\Pi$ 
if there is a sequence $P_1,\dots,P_n$ 
of procedures (possibly empty and possibly with repetitions)  from $\Pi$ such that $Q$ is compatible with and true in each instance 
$I'$ in the set  $\outcome_{P_1,\dots,P_n}(I)$. (If the latter conditions involving $Q$ are true on $I$, then we say that $I$ is ready for $Q$.) We are confident that this problem is decidable for 
sets of procedures in $\P^{\textit{safe},\textit{alter}}$, and we plan on looking into more expressive fragments. 

The proposed framework presents opportunities for several directions of further research. One line of work would involve understanding how to 
represent outcomes of sequences of procedures, or how to obtain good approximations of outcomes of 
more expressive classes of procedures. To solve this problem, we would need a better understanding of the interplay between 
conditional tables and knowledge bases, which would be interesting in its own right. 

We also believe that our framework is general enough to allow reasoning on other data paradigms, or even across various 
different data paradigms. Our black-box abstraction could, for example, offer an effective way to reason about procedures involving unstructured text data, or  
even data transformations using machine-learning tools, as long as one can obtain some guarantees on the data outcomes of these tools. 

{\small 
\bibliographystyle{abbrv}
\bibliography{draft.bib}

\begin{thebibliography}{10}

\bibitem{AHV95}
S.~Abiteboul, R.~Hull, and V.~Vianu.
\newblock {\em Foundations of databases: the logical level}.
\newblock Addison-Wesley Longman Publishing Co., Inc., 1995.

\bibitem{APR13}
M.~Arenas, J.~P{\'e}rez, and J.~Reutter.
\newblock Data exchange beyond complete data.
\newblock {\em Journal of the ACM}, 60(4):28, 2013.

\bibitem{BerardiCGHLM05}
D.~Berardi, D.~Calvanese, G.~{De Giacomo}, R.~Hull, M.~Lenzerini, and
  M.~Mecella.
\newblock Modeling data {\&} processes for service specifications in {C}olombo.
\newblock In {\em Proceedings of the Open Interop. Workshop on Enterprise
  Modelling and Ontologies for Interoperability}, 2005.

\bibitem{BerardiCGHM05}
D.~Berardi, D.~Calvanese, G.~{De Giacomo}, R.~Hull, and M.~Mecella.
\newblock Automatic composition of web services in {C}olombo.
\newblock In {\em Proceedings of the Thirteenth Italian Symposium on Advanced
  Database Systems {(SEBD)}}, pages 8--15, 2005.

\bibitem{BerardiCGLM05}
D.~Berardi, D.~Calvanese, G.~{De Giacomo}, M.~Lenzerini, and M.~Mecella.
\newblock Automatic service composition based on behavioral descriptions.
\newblock {\em Int. J. Cooperative Inf. Syst.}, 14(4):333--376, 2005.

\bibitem{BergmanMNT15}
M.~Bergman, T.~Milo, S.~Novgorodov, and W.~Tan.
\newblock {QOCO:} {A} query oriented data cleaning system with oracles.
\newblock {\em {PVLDB}}, 8(12):1900--1911, 2015.

\bibitem{BergmanMNTsigmod15}
M.~Bergman, T.~Milo, S.~Novgorodov, and W.~C. Tan.
\newblock Query-oriented data cleaning with oracles.
\newblock In {\em Proceedings of {ACM} {SIGMOD}}, pages 1199--1214, 2015.

\bibitem{BGHLS07}
K.~Bhattacharya, C.~Gerede, R.~Hull, R.~Liu, and J.~Su.
\newblock Towards formal analysis of artifact-centric business process models.
\newblock In {\em International Conference on Business Process Management},
  pages 288--304. Springer, 2007.

\bibitem{BO15}
M.~Bienvenu and M.~Ortiz.
\newblock Ontology-mediated query answering with data-tractable description
  logics.
\newblock In {\em Reasoning Web International Summer School}, pages 218--307.
  Springer, 2015.

\bibitem{ChengalurSP98}
I.~Chengalur-Smith and H.~Pazer.
\newblock Decision complacency, consensus and consistency in the presence of
  data quality information.
\newblock In {\em Information Quality}, pages 88--101, 1998.

\bibitem{2012Deutsch}
D.~Deutch and T.~Milo.
\newblock {\em Business Processes: {A} Database Perspective}.
\newblock Synthesis Lectures on Data Management. Morgan {\&} Claypool
  Publishers, 2012.

\bibitem{DHPV09}
A.~Deutsch, R.~Hull, F.~Patrizi, and V.~Vianu.
\newblock Automatic verification of data-centric business processes.
\newblock In {\em Proceedings of the 12th International Conference on Database
  Theory}, pages 252--267. ACM, 2009.

\bibitem{Devlin97}
B.~Devlin.
\newblock {\em Data Warehouse: From Architecture to Implementation}.
\newblock Addison-Wesley Longman, 1996.

\bibitem{FKMP05}
R.~Fagin, P.~G. Kolaitis, R.~J. Miller, and L.~Popa.
\newblock Data exchange: semantics and query answering.
\newblock {\em Theoretical Computer Science}, 336(1):89--124, 2005.

\bibitem{2012FanGeertsBook}
W.~Fan and F.~Geerts.
\newblock {\em Foundations of Data Quality Management}.
\newblock Synthesis Lectures on Data Management. Morgan {\&} Claypool
  Publishers, 2012.

\bibitem{IL84}
T.~Imieli{\'n}ski and W.~Lipski~Jr.
\newblock Incomplete information in relational databases.
\newblock {\em Journal of the ACM (JACM)}, 31(4):761--791, 1984.

\bibitem{KahnSW02}
B.~Kahn, D.~Strong, and R.~Wang.
\newblock Information quality benchmarks: Product and service performance.
\newblock {\em Comm. {ACM}}, 45(4ve):184--192, 2002.

\bibitem{Kimball04}
R.~Kimball and J.~Caserta.
\newblock {\em The Data Warehouse {ETL} Toolkit: Practical Techniques for
  Extracting, Cleaning, Conforming, and Delivering Data}.
\newblock Wiley, 2004.

\bibitem{KindigS03}
D.~Kindig and G.~Stoddart.
\newblock What is population health?
\newblock {\em Am. J. Public Health}, 93(3):380--383, 2003.

\bibitem{KPT06}
P.~G. Kolaitis, J.~Panttaja, and W.-C. Tan.
\newblock The complexity of data exchange.
\newblock In {\em Proceedings of the twenty-fifth ACM SIGMOD-SIGACT-SIGART
  symposium on Principles of database systems}, pages 30--39, 2006.

\bibitem{KrishnanWFGKM015}
S.~Krishnan, J.~Wang, M.~J. Franklin, K.~Goldberg, T.~Kraska, T.~Milo, and
  E.~Wu.
\newblock Sample{C}lean: Fast and reliable analytics on dirty data.
\newblock {\em {IEEE} Data Eng. Bull.}, 38(3):59--75, 2015.

\bibitem{LeePWF09}
Y.~Lee, L.~Pipino, R.~Wang, and J.~Funk.
\newblock {\em Journey to Data Quality}.
\newblock MIT Press, 2009.

\bibitem{LeeS04}
Y.~Lee and D.~Strong.
\newblock Knowing-why about data processes and data quality.
\newblock {\em Journal of Management Information Systems}, 20(3):13--39, 2004.

\bibitem{LeeSKW02}
Y.~Lee, D.~Strong, B.~Kahn, and R.~Wang.
\newblock {AIMQ}: a methodology for information quality assessment.
\newblock {\em Information \& Management}, 40:133--146, 2002.

\bibitem{McAlearney03}
A.~Mc{A}learney.
\newblock {\em Population Health Management: Strategies to Improve Outcomes}.
\newblock Health Administration Press, 2003.

\bibitem{NuttPS15}
W.~Nutt, S.~Paramonov, and O.~Savkovic.
\newblock Implementing query completeness reasoning.
\newblock In {\em {ACM} {CIKM}}, pages 733--742, 2015.

\bibitem{premierWebSite}
Premier, {I}nc.: Alliance of healthcare providers on a mission to transform
  healthcare, 2016.
\newblock {\tt https://www.premierinc.com}.

\bibitem{RazniewskiKNS15}
S.~Razniewski, F.~Korn, W.~Nutt, and D.~Srivastava.
\newblock Identifying the extent of completeness of query answers over
  partially complete databases.
\newblock In {\em {ACM} {SIGMOD}}, pages 561--576, 2015.

\bibitem{Wang98}
R.~Wang.
\newblock A product perspective on total data quality management.
\newblock {\em Comm. {ACM}}, 41(2), 1998.

\bibitem{WangLPS98}
R.~Wang, Y.~Lee, L.~Pipino, and D.~Strong.
\newblock Manage your information as product: The keystone to quality
  information.
\newblock {\em MIT Sloan Management Review}, 39(4):95--105, 1998.

\bibitem{WangS96}
R.~Wang and D.~Strong.
\newblock Beyond accuracy: What data quality means to data consumers.
\newblock {\em Journal of Management Information Systems}, 12(4):5--34, 1996.

\bibitem{Young98}
T.~Young.
\newblock {\em Population Health: Concepts and Methods}.
\newblock Oxford University Press, 1998.

\end{thebibliography}
}

\newpage

\appendix
\section{Additional Examples}
\label{app-exa}

\subsection{SQL data-modification statements}

We show how to encode arbitrary SQL INSERT and DELETE statements as procedures. 
Due to dealing with arbitrary SQL, we relax the constraints and queries that we use. 

\medskip
\noindent
\textbf{INSERT statements}: Consider a SQL statement of the form 
\texttt{INSERT INTO S Q}, where $Q$ is a relational-algebra query.  

\noindent
$\Scope$: Not surprisingly, the scope of the procedure is the relation \texttt{S}. 

\noindent
$\Cpre$: The precondition for the procedure is that all the relation names and attributes mentioned in 
$Q$ must be present in the database. 

\noindent
$\Cpost$: The postcondition is stated using the constraint $\texttt{Q}\subseteq \texttt{S}$.  
(Note that the SQL statement only works when $\texttt{Q}$ and $\texttt{S}$ have the same arity.)  

\noindent
$\Q_\safe$: Since we are inserting tuples, we need the query \texttt{S} to be preserved. 

Alternatively, we can specify an INSERT statement of the form 
\texttt{INSERT INTO S VALUES $\bar a$}, with $\bar a$ a tuple of values. 
In order to formalize this, we just need to change the postcondition of the procedure to 
$\bar a \subseteq \texttt{S}$. 

\medskip
\noindent
\textbf{DELETE statements}: Consider a SQL statement of the form 
\texttt{DELETE FROM S WHERE $C$}, in which $C$ is a boolean combination of conditions. 
 
\noindent
$\Scope$: The scope is the relation $\texttt{S}$, as expected. 

\noindent
$\Cpre$: The precondition for the procedure is that all the relations and attributes mentioned in 
$C$ must be present in the database. 

\noindent
$\Cpost$: There are no postconditions in this query.

\noindent
$\Q_\safe$: Let $Q_C$ be the query \texttt{SELECT * FROM S WHERE C}. Then the safety query is $\texttt{S} - Q_C$, which preserves only those tuples that are not to be deleted.

\subsection{Representing sequences of procedures}
\label{app-more-proofs}

As we mentioned, one possibility to obtain a full representation of sequences of procedures is to 
further restrict the scope of sequences of safe procedures. To be more precise, let us 
say that a sequence $\Pro_1,\dots,\Pro_n$ of procedures is a \emph{safe sequence} if (1) each $\Pro_i$ is either an alter-schema procedure or a safe-scope procedure that only uses tgds, 
and (2) for every $1 \leq  j \leq n$, none of the atoms on the right-hand side of a tgd in $\Pro_j$ is part of the scope of any $\Pro_i$, with $i \leq j$. 
Intuitively, safe sequences of procedures restrict the possibility of sequencing data-migration tasks when the result of one migration is used 
as an input for the next one. 

A conditional instance \emph{with scope} is a pair $\T = (T,\Rel)$, where $T$ is a conditional instance and $\Rel$ is a set of relation names. 
The set of instances represented by $\T$, denoted again by 
$\rep(\T)$, now contains all the instances $J$ in $\rep(T)$ where, 
for each relation $R \in \schema(\nu(T))$ that is not in $\Rel$, 
the projection of $R^J$ over the attributes of $R$ in $T$ is the same as $R^{\nu(T)}$. (In other words, 
we allow extra tuples only in the relations whose symbols are in the set $\Rel$.)  It is now not difficult to show the following result. 

\begin{proposition}
\label{prop-safe-sequence}
For each instance $I$ and each safe sequence $\Pro_1,\dots,\Pro_n$ of procedures one can construct a conditional 
instance $\T$ with scope such that $\rep(\T) = \outcome_{\Pro_1,\dots,\Pro_n}(I)$. 
\end{proposition}

\section{Proofs and Intermediate Results}

\subsection{Proof of Proposition \ref{prop-ap-und}}

The reduction is from the complement of the embedding problem for finite semigroups, shown to be undecidable in \cite{KPT06}, and 
it is itself an adaptation of the proof of Theorem 7.2 in \cite{APR13}.  
Note that, since we do not intend to add attributes nor relations in the procedures of this proof, we can 
drop the named definition of queries, treating CQs now as normal conjunctions of relational atoms. 

The embedding problem for finite semigroups problem 
can be stated as follows. Consider a pair $\textbf A = (A,g)$, where $A$ is a finite set and 
$g: A \times A \rightarrow A$ is a partial associative function. We say that $\textbf A$ is embeddable in a finite 
semigroup is there exists $\textbf B = (B,f)$ such that $A \subseteq B$ and $f: B \times B \rightarrow B$ is a total 
associative function. The embedding problem for finite semigroups is to decide whether an arbitrary 
$\textbf A = (A,g)$  is embeddable in a finite semigroup. 

Consider the schema 
$\Sch = \{C(\cdot,\cdot), E(\cdot,\cdot), N(\cdot,\cdot), G(\cdot,\cdot,\cdot), F(\cdot), D(\cdot)\}$. 
The idea of the proof is as follows. 
We use relation $G$ to encode binary functions, so that a tuple $(a,b,c)$ in $G$ intuitively 
corresponds to saying that $g(a,b) = c$, for a function $g$. 
Using our procedure we  shall mandate that the binary function encoded in $G$ is total and associative. 
We then encode $\textbf A = (A,g)$ into our input instance $I$: the procedure will then try to embed 
$A$ into a semigroup whose function is total. 

In order to construct the procedures, we first specify the following set $\Sigma$ of tgds. 
First we add to $\Sigma$ a set of dependencies ensuring that all elements in the relation 
$G$ are collected into $D$: 

\begin{eqnarray}
G(x,u,v) & \rightarrow & D(x) \\
G(u,x,v) & \rightarrow & D(x) \\
G(u,v,x) & \rightarrow & D(x) 
\end{eqnarray}

The next set verifies that $G$ is total and associative:

\begin{eqnarray}
D(x) \wedge D(y) & \rightarrow & \exists z G(x,y,z) \\
G(x,y,u) \wedge G(u,z,v) \wedge G(y,z,w) & \rightarrow & G(x,w,v)
\end{eqnarray}

Next we include dependencies that are intended to force relation $E$ to be an equivalence relation 
over all elements in the domain of $G$. 

\begin{eqnarray}
D(x) & \rightarrow & E(x,x) \\
E(x,y) & \rightarrow & E(y,x) \\
E(x,y) \wedge E(y,z) & \rightarrow & E(x,z)
\end{eqnarray}

The next set of dependencies we add $\Sigma$ ensure that $G$ represents a function that is consistent 
with the equivalence relation $E$.

\begin{eqnarray}
G(x,y,z) \wedge E(x,x') \wedge E(y,y') \wedge E(z,z') & \rightarrow & G(x',y',z') \\
G(x,y,z) \wedge G(x',y',z') \wedge E(x,x') \wedge E(y,y') & \rightarrow & E(z,z')
\end{eqnarray}

The final tgd in $\Sigma$ serves us to collect possible errors when trying to embed 
$\textbf A = (A,g)$. The intuition for this tgd will be made clear once we outline the reduction, but 
the idea is to state that the relation $F$ now contains everything that is in $R$, as long 
as a certain property holds on relations $E$, $C$ and $N$.  

\begin{eqnarray}
E(x,y) \wedge C(u,x) \wedge C(v,y) \wedge N(u,v) \wedge R(w) & \rightarrow & F(w)
\end{eqnarray}

Let then $\Sigma$ consists of tgds (1)-(11). We construct fixed procedures $P_1 = (\Scope^1,\Cpre^1,\Cpost^1,\Q_\safe^1)$ and 
$P_2= (\Scope^2,\Cpre^2,\Cpost^2,\Q_\safe^2)$ as follows. 

\medskip

\noindent
procedure \textbf{$P_1$}: 

\noindent
$\Scope^1$: The scope of $P_1$ consists of relations $G$, $E$, $D$ and $F$, which corresponds to the constraints 
$\{G[*],  E[*], D[*], F[*]\}$. 

\noindent
$\Cpre^1$: There are no preconditions for this procedure. 

\noindent
$\Cpost^1$: The postconditions are the tgds in $\Sigma$. 

\noindent
$\Q_\safe^1$: This query ensures that no information is deleted from all of $G$, $E$ and $F$ (and thus that no 
attributes are added to them): $G(x,y,z) \wedge E(u,v) \wedge D(w) \wedge F(p)$.

\medskip
\noindent
procedure \textbf{$P_2$}: 

\noindent
$\Scope^2$: The scope of $P_2$ is empty. 

\noindent
$\Cpre^2$: The precondition for this constraint is $R(x) \rightarrow F(x)$. 

\noindent
$\Cpost^2$: The are no postconditions. 

\noindent
$\Q_\safe^2$: There is no safety query. 

\medskip
Note that $P_2$ does not really do anything, it is only there to check that $R$ is contained in $F$. 
We can now state the reduction. On input $\textbf A = (A,g)$, where $A = \{a_1,\dots,a_n\}$, 
we construct an instance 
$I_\textbf{A}$ given by the following interpretations: 
\begin{itemize}
\item $E^{I_\textbf{A}}$ contains the pair $(a_i,a_i)$ for each $1 \leq i \leq n$ (that is, for each element of $A$); 
\item $G^{I_\textbf{A}}$ contains the triple $(a_i,a_j,a_k)$ for each $a_i,a_j,a_k \in A$ such that $g(a_i,a_j) = a_k$;  
\item $D^{I_\textbf{A}}$ and $F^{I_\textbf{A}}$ are empty, while $R^{I_\textbf{A}}$ contains a single element $d$ not in $A$; 
\item $C^{I_\textbf{A}}$ contains the pair $(i,a_i)$ for each $1 \leq i \leq n$; and 
\item $N^{I_\textbf{A}}$ contains the pair $(i,j)$ for each $i \neq j$, $1 \leq i \leq n$ and $1 \leq j \leq n$. 
\end{itemize}

Let us now show $\textbf{A} = (A,g)$ is embeddable in a finite semigroup if and only if $\outcome_{P_1}(I)$ contains 
an instance $I$ such that $I'$ does not satisfy the precondition $R(x) \rightarrow F(x)$ of procedure $P_2$.

\medskip
\noindent
($\Longrightarrow$) Assume that $\textbf{A} = (A,g)$ is embeddable in a finite semigroup, say the semigroup 
$\textbf B = (B,f)$, where $f$ is total. Let $J$ be the instance such that 
$E^J$ is the identity over $B$, $D^J = B$ and $G^J$ contains a pair $(b_1,b_2,b_3)$ if and only if 
$f(b_1,b_2) = b_3$; $F^J$ is empty and relations $N$, $C$ and $R$ are interpreted as in $I_\textbf{A}$. 
It is easy to see that $J \models \Sigma$, $Q_{\Sch \setminus \Scope}$ is preserved and that 
$\Q_\safe(I_\textbf{A}) \subseteq \Q_\safe(J)$, this last because $\textbf A$ was said to be embeddable in 
$\textbf B$. We have that $J$ then does belong to $\outcome_{P_1}(I)$, but $J$ does not satisfy the 
constraint $R(x) \rightarrow F(x)$. 

\medskip
\noindent
($\Longleftarrow$) Assume now that there is an instance $J \in \outcome_{P_1}(I)$ that does not 
satisfy $R(x) \rightarrow F(x)$. Note that, because of the scope of $P_1$, the interpretation of 
$C$, $N$ and $R$ of $J$ must be just as in $I$. Thus it must be that the element $d$ is not in 
$F^J$, because it is the only element in $R^J$. 
Construct a finite semigroup $\textbf{B} = (B,f)$ as follows. 
Let $B$ consists of one representative of each equivalence class in $E^J$, with the additional restriction that 
each $a_i$ in $A$ must be picked as its representative. Further, define $f(b_1,b_2) = b_3$ if and only if 
$G(b_1,b_2,b_3)$ is in $G$. Note that $J$ satisfies the tgds in $\Sigma$, in particular $G$ is associative and 
$E$ acts as en equivalence relation over $G$, which means that $f$ is indeed associative, total, and well defined. 
It remains to show that $\textbf A$ can be embedded in $\textbf B$, but since 
$G^J$ and $E^J$ are supersets of $G^{I_\textbf{A}}$ and $E^{I_\textbf{A}}$ (because of the safety query of $P_1$), 
all we need to show is that each $a_i$ is in a separate equivalence relation. But this hold because of tgd 
(11) in $\Sigma$: if two elements from $A$ are in the same equivalence relation then the left hand side of (11) would hold in $I_\textbf{A}$, which contradicts the fact that $F^J$ does not contain $d$.

\subsection{Proof of Proposition \ref{ref-rep-decidable}}

\newcommand{\Schmin}{\Sch_\text{min}}


Let  $P = (\Scope,\Cpre,\Cpost,\Q_\safe)$. We first show how to construct, for each instance $I$ over a schema $\Sch$, the \emph{minimal} 
schema $\Schmin$ such that all pairs $(J,\Sch')$ that are possible outcomes of applying $P$ over $(I,\Sch)$  are such that 
$\Sch'$ extend $\Schmin$. 

The algorithm receives a procedure $P$ and a schema $\Sch$ and outputs either $\Schmin$, if the procedure is applicable, or a failure signal in case there is no schema satisfying the output constraints of the procedure. 
Along the algorithm we will be assigning numbers to some of the relations in $\Schmin$. This is important to be able to decide failure. 

\medskip
\noindent
\textbf{Algorithm $A(P,\Sch)$ for constructing $\Schmin$}\\
\textbf{Input}: procedure $P = (\Scope,\Cpre,\Cpost,\Q_\safe$ and schema $\Sch$. \\
\textbf{Output}: either \emph{failiure} or a schema $\Schmin$.

\begin{enumerate}
\item If $\Sch$ does not satisfy the structural constraints in $\Cpre$ or is not compatible with either $\Q_\safe$ or $Q_{\Sch \setminus \Scope}$, output failure. Otherwise, continue. 
\item Start with $\Schmin = \emptyset$. 
\item For each total query $R$ in $\Q_\safe$, assume that $|\Sch(R)| = k$. Set $\Schmin(R) = \Sch(R)$, and label $R$ with $k$.  

\item Add to $\Schmin$ all relations $R$ mentioned in an atom $R[*]$ in $\Cpost$ (if they are not already part of $\Schmin$), without associating any attributes to them 

\item In the following instructions we construct a set $\Gamma(P,\Sch)$ of pairs of relations and attributes. Intuitively, 
a pair $(R,\{a_1,\dots,a_n\})$ in $\Gamma(P,\Sch)$ states that each schema in the output of $P$ must contain a relation $R$ with attributes 
$a_1,\dots,a_n$. 
\begin{itemize}
\item For each relation $R$ in $\Sch$ that is not mentioned in $\Scope$, add to $\Gamma(P,\Sch)$ the pair $(R,\Sch(R))$.
\item For each constraint $R[a_1,\dots,a_n]$ in $\Scope$, add the pair 
$(R,\Sch(R) \setminus \{a_1,\dots,a_n\})$ to $\Gamma(P,\Sch)$.
\item For each atom $R(a_1:x_1,\dots,a_n:x_n)$ in $\Q_\safe$ add to $\Gamma(P,\Sch)$ the pair 
$(R,\{a_1,\dots,a_n\})$. 
\item For each atom $R(a_1:x_1,\dots,a_n:x_n)$ in a tgd or egd in $\Cpost$ add to $\Gamma(P,\Sch)$ the pair 
$(R,\{a_1,\dots,a_n\})$. 
\item For each constraint $R[a_1,\dots,a_n]$ in $\Cpost$, add to $\Gamma(P,\Sch)$ the pair 
$(R,\{a_1,\dots,a_n\})$. 
\end{itemize}

\item For each pair $(R,A)$ in $\Gamma(P,\Sch)$, do the following.
\begin{itemize}
\item If $R$ is not yet in $\Schmin$, add $R$ to $\Schmin$ and set $\Schmin(R) = A$; 
\item If $R$ is in $\Schmin$, update $\Schmin(R) = \Schmin(R) \cup A$. 
\end{itemize}


\item If $\Schmin$ contains a relation $R$ labelled with a number $n$ where, $\Schmin(R) > n$, output failure. Otherwise output $\Schmin$.  
\end{enumerate}

By direct inspection of the algorithm, we can state the following. 
\begin{proposition}
\label{obs-gamma}
Let $P = (\Scope,\Cpre,\Cpost,\Q_\safe)$ be a relational procedure and $\Sch$ a relational schema. Then for each 
relation $R$ in $\Schmin$ with attributes $\{a_1,\dots,a_n\}$, every instance $I$ over $\Sch$ and every pair 
$(J,\Sch')$ in the outcome of applying $P$ to $(I)$, we have that 
$\Sch(R)$ is defined, with $\{a_1,\dots,a_n\} \subseteq \Sch(R)$. 
\end{proposition}


Furthermore, the following lemma specifies, in a sense, the correctness of the algorithm. 
\begin{lemma}
\label{lem-schmin}
Let $P = (\Scope,\Cpre,\Cpost,\Q_\safe)$ be a relational procedure and $\Sch$ a relational schema. Then: 
\begin{itemize}
\item[i)] If $A(P,\Sch)$ outputs failure, either $P$ cannot be applied over any instance $I$ over $\Sch$, or 
for each instance $I$ over $\Sch$ the set  $\outcome_P(I)$ is empty. 
\item[ii)] If $A(P,\Sch)$ outputs $\Schmin$, then the schema of any instance in $\outcome_P(I)$ extends $\Schmin$.
\end{itemize}
\end{lemma}
\begin{proof}
For i), not that if some of the components of $P$ are not compatible with $\Sch$, or $\Sch$ does not satisfy the constraints in 
$\Cpre$, then clearly $P$ cannot be applied over any instance $I$ over $\Sch$. 
Assume then that $\Sch$ satisfies all compatibilities and preconditions in $P$, but $A(P,\Sch)$ outputs failure. 
Then $\Schmin$ contains a relation $R$ such that $|\Schmin(R)| = m$, but $R$ is labelled with number $k$, for 
$k < \ell$. From the algorithm, we this implies that $|\Schmin(R)| > |\Sch(R)$, but that there is a query 
$R$ in $\Q_\safe$. Clearly, $\Q_\safe$ cannot be preserved under any outcome, since by Observation \ref{obs-gamma} 
we require the schemas of outcomes to assign more attributes to $R$ than those assigned by $\Schmin$, and thus 
the cardinality of tuples in the answer of $R$ differs between $I$ and its possible outcomes. 
Finally, item ii) is a direct consequence of Observation \ref{obs-gamma}. 
\end{proof}

Note that the algorithm $(A,P)$ runs in polynomial time, and that the total size of $\Schmin$ (measured as the number of relations and attributes) 
is at most the size of $\Sch$ and $P$ combined. Thus, to decide the applicability problem for a sequence $P_1,\dots,P_n$ of procedures, all we need 
to do is to perform subsequent calls to the algorithm, setting $\Sch_0 = \Sch$ and then using $\Sch_{i} =  A(P_i,\Sch_{i-1})$ as the input for the 
next procedures. If $A(P_n,\Sch_{n-1})$ outputs a schema, then the answer to the applicability problem is affirmative, otherwise if some call to 
$A(P_i,\Sch{i-1})$ outputs failure, the 
answer is negative.

\subsection{Proof of Proposition \ref{prop-rep-undec}}

This proof is a simple adaptation of the reduction we used in the proof of Proposition \ref{prop-ap-und}. 
Indeed, consider again the schema $\Sch$ from this proof, and the procedure $P$ given by: 

\noindent
$\Scope$: The scope of $P$ consists of relations $G$, $E$, $D$ and $F$, which corresponds to the constraints 
$G[*], E[*], D[*]$ and $F[*]$. 

\noindent
$\Cpre$: There are no preconditions for this procedure. 

\noindent
$\Cpost$: The postconditions are the tgds in $\Sigma$ plus the tgd $F(x) \rightarrow R(x)$.

\noindent
$\Q_\safe$: This query ensures that no information is deleted from all of $G$, $E$ and $F$: 
$G(x,y,z) \wedge E(u,v) \wedge D(w) \wedge F(p)$.

Given a finite semigroup $\textbf A$, we construct now the following instance $I$: 
\begin{itemize}
\item $E^{I_\textbf{A}}$ contains the pair $(a_i,a_i)$ for each $1 \leq i \leq n$ (that is, for each element of $A$); 
\item $G^{I_\textbf{A}}$ contains the triple $(a_i,a_j,a_k)$ for each $a_i,a_j,a_k \in A$ such that $g(a_i,a_j) = a_k$;  
\item All of $D^{I_\textbf{A}}$, $F^{I_\textbf{A}}$ and $R^{I_\textbf{A}}$ are empty; 
\item $C^{I_\textbf{A}}$ contains the pair $(i,a_i)$ for each $1 \leq i \leq n$; and 
\item $N^{I_\textbf{A}}$ contains the pair $(i,j)$ for each $i \neq j$, $1 \leq i \leq n$ and $1 \leq j \leq n$. 
\end{itemize}

By a similar argument as the one used in the proof of Proposition \ref{prop-ap-und}, one can show that 
$\outcome_P(I)$ has an instance if and only if $\textbf A$ is embeddable in a finite semigroup. 
The intuition is that now we are adding the constraint $F(x) \rightarrow R(x)$ as a postcondition, and 
since $R$ is not part of the scope of the procedure the only way to satisfy this restriction is if we do not 
fire the tgd (11) of the set $\Sigma$ constructed in the aforementioned proof. This, in turn, can only happen if 
$\textbf A$ is embeddable. 

\subsection{Proof of Proposition \ref{prop-rep-undec-nothingworks}}

The reduction, just as that of Proposition \ref{prop-ap-und}, is by reduction from 
the embedding problem for finite semigroups, and builds up from this proposition. 
Let us start by defining the procedures $\Pro_1$, $\Pro_2$ and $\Pro_3$. 

\medskip
\noindent
For procedure $\Pro_1$ we first build a set $\Gamma_1$ of tgs. This set is similar to the set 
$\Sigma$ used in Proposition \ref{prop-ap-und}, but using three additional \emph{dummy} 
relations $G^d$, $E^d$ and $G^\text{binary}$. 

First we add to $\Gamma_1$ dependencies that collect elements of $G$ into $D$, and that initialize 
$E$ as a reflexive relation. 

\begin{eqnarray*}
G(x,u,v) & \rightarrow & D(x) \\
G(u,x,v) & \rightarrow & D(x) \\
G(u,v,x) & \rightarrow & D(x) \\ 
D(x) & \rightarrow & E(x,x) 
\end{eqnarray*}

Next the dependency that states that $F$ contains everything in $R$ if some conditions about $E$ occur. 

\begin{eqnarray}
E(x,y) \wedge C(u,x) \wedge C(v,y) \wedge N(u,v) \wedge R(w) & \rightarrow & F(w)
\end{eqnarray}

The dependencies that assured that $E$ was an equivalence relation where acyclic, so we 
replace the right hand side with a dummy relation.  

\begin{eqnarray*}
E(x,y) & \rightarrow & E^d(y,x) \\
E(x,y) \wedge E(y,z) & \rightarrow & E^d(x,z)
\end{eqnarray*}

Next come the dependencies assuring $G$ is a total and associative function, using also dummy relations. 

\begin{eqnarray*}
D(x) \wedge D(y) & \rightarrow & \exists z G^\text{binary}(x,y) \\
G(x,y,u) \wedge G(u,z,v) \wedge G(y,z,w) & \rightarrow & G^d(x,w,v)
\end{eqnarray*}

Finally, the dependencies that were supposed to ensure that $E$ worked as the equality over 
function $G$, using again the dummy relations.  

\begin{eqnarray*}
G(x,y,z) \wedge E(x,x') \wedge E(y,y') \wedge E(z,z') & \rightarrow & G^d(x',y',z') \\
G(x,y,z) \wedge G(x',y',z') \wedge E(x,x') \wedge E(y,y') & \rightarrow & E^d(z,z')
\end{eqnarray*}

\noindent
We can now define procedure \textbf{$P_1$}: 

\noindent
$\Scope$: The scope of $P_1$ consists of relations $G$, $E$, $D$, $F$, $G^d$, $E^d$ and $G^\text{binary}$ which corresponds to the constraints 
$G[*], E[*], D[*], F[*], E^d[*], G^d[*]$ and $G^\text{binary}[*]$. 

\noindent
$\Cpre$: There are no preconditions for this procedure. 

\noindent
$\Cpost$: The postconditions are the tgds in $\Gamma_1$. 

\noindent
$\Q_\safe$: This query ensures that no information is deleted from all of $G$, $E$, $F$, $G^d$, $E^d$ and $G^\text{binary}$: 
$G(x,y,z) \wedge E(u,v) \wedge D(w) \wedge F(p) 
\wedge G^d(x',y',z') \wedge E^d(u',v') \wedge G^\text{binary}(a,b)$.

Note that, even though relations $G$ and $E$ are not mentioned in the right hand side of any tgd in $\Gamma_1$, 
they are part of the scope and thus they could be modified by the procedures $\Pro_1$. 

\medskip
\noindent
The procedure $\Pro_2$ has no scope, no safety queries, no precondition, and the only postcondition is the presence 
of a third attribute, say $C$, in $G^\text{binary}$, by using a structural constraint $G^\text{binary}[A,B,C]$ (to maintain consistency with our unnamed 
perspective, we assume that these three attributes are ordered $A <_\mathcal{A} B <_\mathcal{A} C$). 

\medskip
\noindent
To define the final procedure, consider the following set of tgds $\Gamma_3$. 

\begin{eqnarray*}
E^d(x,y) & \rightarrow & E(x,y) \\ 
G^d(x,y,z) & \rightarrow & G(x,y,z) \\  
G^\text{binary}(x,y,z) & \rightarrow & G(x,y,z) \\ 
F(x) & \rightarrow & F^\text{check}(x) 
\end{eqnarray*}

\noindent
Then we define procedure $\Pro_3$ is as follows. 

\noindent
$\Scope$: The scope of $\Pro_3$ is again empty.  

\noindent
$\Cpre$: There are no preconditions for this procedure. 

\noindent
$\Cpost$: The postconditions are the tgds in $\Gamma_3$. 

\noindent
$\Q_\safe$: There are also no safety queries for this procedure. 

Let $\Sch$ be the schema containing relations 
$G$, $E$, $D$, $F$, $F^\text{check}$, $G^d$, $E^d$ and $G^\text{binary}$ and $R$. The attribute names are of no importance 
for this proof, except for $G^\text{binary}$, which associates attributes $A$ and $B$. 

Given a finite semigroup $\textbf A$, we construct now the following instance $I_{\textbf A}$: 
\begin{itemize}
\item $E^{I_\textbf{A}}$ contains the pair $(a_i,a_i)$ for each $1 \leq i \leq n$ (that is, for each element of $A$); 
\item $G^{I_\textbf{A}}$ contains the triple $(a_i,a_j,a_k)$ for each $a_i,a_j,a_k \in A$ such that $g(a_i,a_j) = a_k$;  
\item All of $D^{I_\textbf{A}}$, $F^{I_\textbf{A}}$ and ${F^\text{check}}^{I_\textbf{A}}$ are empty; 
\item $R^{I_\textbf{A}}$ has a single element $d$ not used elsewhere in $I_\textbf{A}$
\item $C^{I_\textbf{A}}$ contains the pair $(i,a_i)$ for each $1 \leq i \leq n$; and 
\item $N^{I_\textbf{A}}$ contains the pair $(i,j)$ for each $i \neq j$, $1 \leq i \leq n$ and $1 \leq j \leq n$. 
\end{itemize}

Let us now show $\textbf{A} = (A,g)$ is embeddable in a finite semigroup if and only if $\outcome_{P_1,P_2,P_3}(I)$ is nonempty. 

\medskip
\noindent
($\Longrightarrow$) Assume that $\textbf{A} = (A,g)$ is embeddable in a finite semigroup, say the semigroup 
$\textbf B = (B,f)$, where $f$ is total. Let $J$ be the instance over $\Sch$ such that 
both ${E^d}^J$ and $E^J$ are the identity over $B$, $D^J = B$, both ${G^d}^J$ and  $G^J$ contains a pair $(b_1,b_2,b_3)$ if and only if 
$f(b_1,b_2) = b_3$; ${G^\text{binary}}^J$ is the projection of $G^J$ over its two first attirbutes, 
$F^J$ and ${F^\text{check}}^J$ are empty and relations $N$, $C$ and $R$ are interpreted as in $I_\textbf{A}$. 

It is easy to see that $J$ is in the outcome of applying $\Pro_1$ over $I$. 
Now, let $\Sch'$ be the extension of $\Sch$ where $G^\text{binary}$ has an extra attribute, $C$, 
and $K$ is an instance over $\Sch'$ that is just like $J$ except that ${G^\text{binary}}^K$ is now 
the same as $G^J$ (and therefore $G^K$). 
By definition we obtain that $K$ is a possible outcome of applying $\Pro_2$ over $J$, and therefore 
$K$ is in $\outcome_{\Pro_1,\Pro_2}(I)$. 
Furthermore, one can see that the same instance $K$ is again an outcome of applying $\Pro_3$ over 
$K$, therefore obtaining that $\outcome_{\Pro_1,\Pro_2,\Pro_3}(I)$ is nonempty. 

\medskip
\noindent
($\Longleftarrow$) Assume now that there is an instance $L \in \outcome_{P_1,P_2,P_3}(I)$. Then by definition there are instances $J$ and $K$ 
such that $J$ is in $\outcome_{\Pro_1}(I)$, $K$ is in $\outcome_{\Pro_2}(J)$ and 
$L$ is in $\outcome_{\Pro_3}(K)$. 

Let $J^*$ be the restriction of $J$ over the schema $\Sch$. From a simple inspection of $\Pro_1$ we have 
that $J^*$ satisfies as well the dependencies in $\Pro_1$, so that 
$J^*$ is in $\outcome_{\Pro_1}(I)$.

Let now $\Sch'$ be the extension of $\Sch$ that assigns also attribute $C$ to $G^\text{binary}$. Now, since $K$ is an outcome of $P_2$ over $J$ 
and $P_2$ has no scope, if we define $K^*$ as the restriction of $K$ over $\Sch'$, then 
clearly $K^*$ must be in the outcome of applying $\Pro_2$ over $J^*$. 
Note that, by definition of $\Pro_3$ (since its scope is empty), the restriction of $L$ up to the schema of $K$ must be the same instance as $K$, 
and therefore the restriction $L^*$ of $L$ to $\Sch'$ must be the same instance than $K^*$. Furthermore, since $L$ (and thus $L^*$) satisfies the 
constraints in $\Pro_3$, and the constraints only mention relations and atoms in $\Sch'$, 
we have that $K^*$ must be an outcome of applying $\Pro_3$ over $(K^*,\Sch')$. 

We now claim that $K^*$ satisfy all tgds (1)-(11) in the proof of Propositon \ref{prop-ap-und}. 
Tgds (1-3) and (6) are immediate from the scopes of procedures, and the satisfaction for all the remaining ones 
is shown in the same way. For example, to see that $K^*$ satisfies $E(x,y) \rightarrow E(y,x)$, note that 
$J^*$ already satisfies $E(x,y) \rightarrow E^d(y,x)$. From the fact that the interpretations of $E^d$ and $E$ 
are the same over $J^*$ and $K^*$ and that $K^*$ satisfies $E^d(x,y) \rightarrow E(x,y)$ we 
obtain the desired result. 

Finally, since $K^*$ satisfies $F(x)\ \rightarrow\ F^\text{check}(x)$, and the interpretation of 
$F^\text{check}$ over all of $I$, $J^*$ and $K^*$ must be empty, we have that 
the interpretation of $F$ over $K^*$ is empty as well. Given that $K^*$ satisfies all dependencies in $\Sigma$, 
it must be the case that the left hand side of the tgd (11) is not true $K^*$, for any possible assignment. 
By using the same argument 
as in the proof of Proposition \ref{prop-ap-und} we obtain that $\textbf{A} = (A,g)$ is embeddable in a finite semigroup.

\subsection{Proof of Theorem \ref{theo-rep-safe-scope}}

This theorem is an immediate corollary of Proposition \ref{prop-minimal}, together with an inspection on the complexity of 
computing the over-approximation. We provide all details in the proof of the next proposition (Proposition \ref{prop-minimal}). 

\subsection{Proof of Proposition \ref{prop-minimal}}

For the proof we assume that all procedures does not use preconditions. We can treat them  by first doing an initial check on compatibility that only complicates the proof. 

We also specify an alternative set of representatives for conditional instances (which is actually the usual one). The 
set $\hat \rep(G)$ of representatives of a conditional instance $G$ is simply 
$\hat \rep(G)= \{I \mid$ there is a substitution $\nu$ such that $\nu(T) \subseteq I \}$. That is,  $\hat \rep (G)$ only 
specifies instances over the same schema as $G$. 
The following lemma allows us to work with this representation instead; it is immediate from the definition of safe scope procedures. 

\begin{lemma}
\label{lem-usual-rep}
If $G$ is a conditional instance, then (1) $\hat \rep(G) \subseteq \rep(G)$, and (2) an instance $J$ is minimal for $\rep(G)$ if and only if it is minimal 
for $\hat \rep(G)$. 
\end{lemma}

Moreover, from the fact that procedures with safe scope are acyclic, we can state Theorem 5.1 in \cite{APR13} in the following terms:  

\begin{lemma}[\cite{APR13}]
\label{lem-chase}
Given a set $\Sigma$ of tgds and a positive conditional instance $G$, one can construct, in polynomial time, a positive conditional instance 
$G'$ such that (1) $\hat \rep(G') \subseteq \hat \rep(G)$ and (2) all minimal models of $\hat \rep(G')$ satisfy $\Sigma$. 
\end{lemma}

Moreover, by slightly adapting the proof of Proposition 4.6 in \cite{APR13}, we can see that the conditional instance constructed above has even better 
properties. In order to prove this theorem all that one needs to do is to adapt the notion of solutions for data exchange into a scenario where 
the target instance may already have some tuples (which will not fire any dependencies because of the safeness of procedures). 

\begin{lemma}[\cite{APR13}]
\label{lem-minimal-chase}
Let $P = \Scope,\Cpre,\Cpost,\Q_\safe$ be a procedure with safe scope, and let $G$ be a positive conditional instance. 
Then one can construct a conditional instance $G'$ such that, for every minimal instance $I$ of $\hat \rep(G)$, the 
set $\hat \rep(G')$ contains all minimal instances in $\outcome_{\Pro}(I)$, and for every minmal instance $J$ in $\hat \rep(G')$ 
there is a minimal instance $I$ of $\rep(G)$ such that $J$ is minimal for $\outcome_{\Pro}(I)$.
\end{lemma}

Finally, we can show the key result for this proof. 
\begin{lemma}
\label{lem-seq-safe}
Let $\Inst$ be a set of instances, and $G$ a conditional table that is minimal for $\Inst$, and 
$\Pro  = (\Scope,\Cpre,\Cpost,\Q_\safe)$ a procedure with safe scope.  
Then either $\outcome_{\Pro}(\Inst) = \emptyset$ or one can construct, in polynomial time, a  
conditional instance $G'$ 
such that 
\begin{itemize}
\item[i)] $\outcome_{\Pro}(\Inst) \subseteq \rep(G')$; and 
\item[ii)] If $J$ is a minimal instance in $\rep(G')$, then $J$ is also minimal in 
$\outcome_{\Pro}(\Inst) $. 
\end{itemize}
\end{lemma}

\begin{proof}
Using the chase procedure mentioned in Lemma \ref{lem-minimal-chase}, we see that the conditional table $G'$ produced in this lemma 
satisfies the conditions of this Lemma, for $\hat \rep(G)$.  

For i), let $J$ be an instance in $\outcome_{\Pro}(\Inst)$. Then there is an instance $I$ in $\Inst$ such that 
$J \in \outcome{P}(I)$. Let $I^*$ be a minimal instance in $\Inst$ such that $I$ extends $I^*$. 
By our assumption we know that $I^*$ belongs to $\rep(G)$, and 
since $I^*$ is minimal it must be the case that $I^*$ belongs (and is minimal) for $\hat \rep(G)$. 
Therefore, by Lemma \ref{lem-minimal-chase} we have that $\hat \rep(G')$ contains all minimal instances 
for $\outcome_P(I^*)$. But now notice that for every assignment $\tau$ and tgd $\lambda$ such that 
$(I^*,\tau)$ satisfies $\lambda$, we have that $(I,\tau)$ satisfy $\lambda$ as well. This means that 
every instance in the set $\outcome_P(I)$ must extend a minimal instance in $\outcome_P(I^*)$ 
(if not, then a tgd would not be satisfied due to some assignment that would not be possible to extend). 
Since every minimal instance in $\outcome_P(I^*)$ is in $\hat \rep(G')$, then by the semantics of conditional tables 
it must be the case that $J$ belongs to $\hat \rep(G')$ as well, and therefore to $\rep(G')$. 

Item [ii)] follows from the fact that any minimal instance in $\rep(G')$ must also be minimal for $\hat \rep(G')$ 
and a direct application of Lemma \ref{lem-minimal-chase}.
\end{proof}

The next Lemma constructs the desired outcomes for alter schema procedures. 

\begin{lemma}
\label{lem-seq-alter}
Let $\Inst$ be a set of instances, and $G$ a conditional table that is minimal for $\Inst$, and 
$\Pro  = (\Scope,\Cpre,\Cpost,\Q_\safe)$ an alter schema procedure.  
Then either $\outcome_{\Pro}(\Inst) = \emptyset$ or one can construct, in polynomial time, a  
conditional instance $G'$ 
such that 
\begin{itemize}
\item[i)] $\outcome_{\Pro}(\Inst) \subseteq \rep(G')$; and 
\item[ii)] If $J$ is a minimal instance in $\rep(G')$, then $J$ is also minimal in 
$\outcome_{\Pro}(\Inst) $. 
\end{itemize}
\end{lemma}

\begin{proof}
Assume that $\outcome_{\Pro}(\Inst) \neq \emptyset$ (this can be easily checked in polynomial time). 
Then one can compute the schema $\Schmin$ from the proof of Proposition 
\ref{ref-rep-decidable}. This schema will add some attributes to some relations in the schema of $G$, and 
possibly some other relations with other sets of attributes. Let $\schema(G) = \Sch$. 

We extend $G$ to a conditional table $G'$ over $\Schmin$ as follows: 
\begin{enumerate}
\item For every relation $R$ such that $\Schmin(R) \setminus \Sch(R) = \{A_1,\dots,A_n\}$, with $n \geq 1$, for tuples from $G'$ by 
adding to each tuple in $G$ a fresh null value in each of the attributes $A_1,\dots,A_n$. 
\item For every relation $R$ such that $\Sch(R)$ is not defined, but $\Schmin(R)$ is defined, set $R^{G'} = \emptyset$
\end{enumerate}

The properties of the lemma now follow from a straightforward check. 
\end{proof}

The proof of Proposition \ref{prop-minimal} now follows from successive applications of Lemmas \ref{lem-seq-alter} and  \ref{lem-seq-safe}: 
one just need to compute the appropriate conditional table for each procedure in the sequence $P_1,\dots,P_n$. 
That each construction is in polynomial size if the number $n$ of procedures is fixed, or exponential in other case, 
follows also from these Lemmas, as the size of the conditional table $G'$, for a procedure $P$ and a conditional table $G$, is at most 
polynomial in $G$ and $P$ (and thus we are composing a polynomial number of polynomials, or a fixed number if $n$ is fixed).

\medskip
\noindent
\textbf{Proof of Theorem \ref{theo-rep-safe-scope}}: 
While in general checking that the set represented by an arbitrary conditional instance may be np-complete, we note that 
in \cite{APR13} it was shown that, for Lemma \ref{lem-seq-safe}, all that is needed is a positive conditional instance, and clearly 
deciding whether a positive conditional instance represents at least one solution is in polynomial time. 
Thus, for the proof of the Theorem we just compute the (positive) conditional instance exhibited for Proposition \ref{prop-minimal} and 
then do the polynomial check on the size of the final conditional instance.

\end{document}